\documentclass[a4paper,12pt]{article}

\usepackage[a4paper,hmargin={2.5cm,1.5cm},vmargin={4cm,3cm}]{geometry}
\usepackage{makecell}
\usepackage[english]{babel}
\usepackage{eurosym}
\usepackage[utf8]{inputenc}
\usepackage{color}
\usepackage{indentfirst}
\usepackage{amsmath}
\usepackage{graphicx}
\usepackage{float}
\usepackage{rotfloat}
\usepackage{rotating}
\usepackage{adjustbox}
\usepackage{multirow}
\usepackage{amssymb}
\usepackage{indentfirst}
 \usepackage{dsfont}
 \usepackage{textcomp}
 \usepackage[T1]{fontenc}
\usepackage{lmodern}
\usepackage{pifont}
\usepackage{graphicx}

\usepackage[ruled,vlined,linesnumbered]{algorithm2e}

\newcommand{\I}{{\mathcal{I}}}

\newcommand{\mar}{{\text{Mar}}}
\newcommand{\ffrac}{{\text{frac}}}
\newcommand{\rank}{{\text{rank}}}

\newcommand{\R}{\mathbb{R}}

\newcommand{\N}{\mathcal{N}}
\newcommand{\eps}{\varepsilon}

\newcommand{\OPT}{\mathrm{OPT}}
\newcommand{\supp}{\mathrm{supp}}

\newcommand{\NN}{\mathcal{N}}

\newcommand{\Mar}{\operatorname{Mar}}

\newcommand{\symy}{{\mathbf{y}}}
\newcommand{\symx}{{\mathbf{x}}}
\newcommand{\symz}{{\mathbf{Z}}}
\newcommand{\reali}{{\text{R}}}

\newcommand{\Rank}{\mathrm{rank}}
\usepackage{amsmath}
\usepackage{amsthm}
\usepackage{amssymb}
\usepackage{textcase} 
\usepackage{dsfont}
\newtheorem{theorem}{Theorem} 
\newtheorem{lemma}[theorem]{Lemma}
\newtheorem{observation}[theorem]{Observation}
\newtheorem{definition}[theorem]{Definition}
\newtheorem{corollary}[theorem]{Corollary}
\usepackage{mathtools}

\usepackage{amsmath}
\usepackage{amsthm}
\usepackage{multirow}
\usepackage{xcolor}
\usepackage{amssymb}
\usepackage{hyperref}
\usepackage{float} 
\newcommand{\lovasz}{{Lov\'{a}sz} }
\usepackage{tikz}

\usepackage{authblk}

\title{Deterministic Algorithm for  Non-monotone Submodular Maximization under Matroid and Knapsack Constraints\thanks{Alphabetical Order. Corresponding Author: Shengminjie Chen, csmj@ict.ac.cn}}
\author[1,2]{Shengminjie Chen}
\author[1,2]{Yiwei Gao}
\author[1,2]{Kaifeng Lin}
\author[1,2]{Xiaoming Sun}
\author[1,2]{Jialin Zhang}
\affil[1]{State Key Lab of Processors, Institute of Computing Technology, Chinese Academy of Sciences, Beijing 100190, China}
\affil[2]{School of Computer Science and Technology, University of Chinese Academy of Sciences, Beijing 100049, China}

 \long\def\symbolfootnote[#1]#2{\begingroup%
\def\thefootnote{\fnsymbol{footnote}}\footnote[#1]{#2}\endgroup}

\begin{document}



\maketitle

\begin{abstract}
    Submodular maximization constitutes a prominent research topic in combinatorial optimization and theoretical computer science, with extensive applications across diverse domains. While substantial advancements have been achieved in approximation algorithms for submodular maximization, the majority of algorithms yielding high approximation  guarantees are randomized. In this work, we investigate deterministic approximation algorithms for maximizing non-monotone submodular functions subject to matroid and knapsack constraints. For the two distinct constraint settings, we propose novel deterministic algorithms grounded in an extended multilinear extension framework. For matroid constraints, our algorithm achieves an approximation ratio of $(0.385 - \eps)$, while for knapsack constraints, the proposed algorithm attains an approximation ratio of $(0.367 -\eps)$. Both algorithms run in $\mathrm{poly}(n)$ value queries, where $n$ is the size of the ground set, and improve upon the state-of-the-art deterministic approximation ratios of $(0.367 - \eps)$ for matroid constraints and $0.25$ for knapsack constraints.
\end{abstract}

\section{Introduction}
Submodularity, capturing the principle of diminishing returns, naturally arises in many functions of theoretical interest and provides a powerful abstraction for many subset selection problems such as game theory  \cite{DBLP:conf/stoc/Dobzinski11,DBLP:conf/soda/DobzinskiS06,DBLP:journals/siamco/DianettiF20}, maximal social welfare \cite{DBLP:conf/stoc/FengH0025, DBLP:conf/focs/LiV21,DBLP:journals/siamcomp/KorulaMZ18}, influence maximization \cite{DBLP:conf/stoc/MosselR07, DBLP:conf/icml/ChenSZZ21, DBLP:conf/kdd/KempeKT03}, facility location \cite{DBLP:journals/dam/AgeevS99, DBLP:journals/orl/Ageev02}, data summarization \cite{DBLP:conf/iccv/SimonSS07, DBLP:conf/cikm/SiposSSJ12, DBLP:conf/nips/TschiatschekIWB14}, etc. Maximizing submodular functions subject to various constraints, due to the $\mathcal{NP}$-hard property, is a fundamental problem that has been extensively studied in both operations research and theoretical computer science, playing a central role in combinatorial optimization and approximation algorithms. Formally, for a ground set $\N$, the submodular function $f: 2^{\N} \rightarrow \R$ defined on $\N$ satisfies the following property: for any subset $A \subseteq B \subseteq \N$ and any element $u \in \N \setminus B$,
\[
f(A \cup \{u\}) - f(A) \ge f(B \cup \{u\}) - f(B).
\]
Equivalently, a function $f$ is submodular if and only if for any sets $S, T \subseteq \N$,
\[ 
f(S) + f(T) \ge f(S \cup T) + f(S \cap T).
 \]

The systematic study of submodular maximization dates back to 1978, when Nemhauser et al. \cite{DBLP:journals/mp/NemhauserWF78} abstracted the concept of submodularity from various combinatorial problems and established the $1 - 1/e$ approximation ratio for monotone functions via a greedy algorithm. Since then, numerous combinatorial algorithms have been developed, particularly for
cardinality constraints~\cite{DBLP:journals/iandc/Skowron17,DBLP:journals/mor/NemhauserW78,DBLP:conf/soda/BuchbinderFNS14},
knapsack constraints~\cite{DBLP:journals/orl/Sviridenko04,DBLP:journals/mor/KulikST13,DBLP:conf/nips/AmanatidisFLLR20,DBLP:journals/algorithmica/FeldmanNS23,DBLP:journals/orl/KulikSS21},
matroid constraints~\cite{Fisher1978,DBLP:journals/siamcomp/BuchbinderFG23,DBLP:conf/icalp/HenzingerLVZ23},
and the unconstrained case~\cite{DBLP:journals/siamcomp/BuchbinderFNS15,DBLP:conf/nips/PanJGBJ14}.
 A significant advancement came with the introduction of the multilinear extension framework by Vondrák \cite{DBLP:conf/stoc/Vondrak08}. This framework extends a discrete submodular function to a continuous multilinear function $F_{ME} : [0,1]^{\N} \to \mathbb{R}$,
\[
F_{ME}(\mathbf{x}) = \mathbb{E}[f(R(\mathbf{x}))] = \sum_{S\subseteq \N} f(S) \prod_{i \in S} x_i \prod_{j \notin S} (1-x_j),
\]
where $R(\mathbf{x})$ is a random subset of $\N$ in which each element $i \in \N$ is included independently with probability $x_i$. By applying a Frank-Wolfe-like continuous optimization algorithm, a fractional solution is obtained, which is then converted into a discrete solution via a rounding procedure such as Pipage Rounding \cite{DBLP:journals/siamcomp/CalinescuCPV11} and Swap Rounding \cite{DBLP:conf/focs/ChekuriVZ10}. This approach achieves an approximation ratio $1 - 1/e$ for monotone functions under matroid constraints. Crucially, the optimization component of this framework generally requires only that the constraint be a down-closed convex polytope. Given that lossless rounding techniques exist for cardinality and matroid constraints,
and that knapsack constraints admit nearly lossless rounding via
contention resolution schemes~\cite{DBLP:journals/siamcomp/ChekuriVZ14},
the multilinear extension framework supports a versatile algorithmic design
that is not tied to specific constraint types. Subsequent research has continuously advanced the approximation ratio for non-monotone submodular maximization, from the long-standing $1/e( \approx 0.367) $ \cite{DBLP:conf/stoc/Vondrak08} to $(0.372 - \eps)$ \cite{DBLP:conf/focs/EneN16}, $(0.385 - \eps)$ \cite{DBLP:journals/mor/BuchbinderF19}, and $(0.401 - \eps)$ \cite{DBLP:conf/stoc/BuchbinderF24}. Closing the remaining gap between these results and the theoretical upper bound of $0.478$ \cite{DBLP:conf/soda/GharanV11} remains a significant open challenge in the field.
Computing the value of the multilinear extension exactly typically requires exponentially many function evaluations, and random sampling is widely regarded as the only polynomial-query approach for estimating its value. However, such randomness can be undesirable in applications where deterministic outcomes are required, such as safety-critical or reproducibility-sensitive environments. Randomized algorithms provide  guarantees in expectation but not necessarily in the worst case, making them less predictable in certain scenarios. 

Designing deterministic algorithms for submodular maximization, which returns the same solution on every run and guarantees an approximation ratio in the worst case rather than only in expectation, has emerged as an important and active research direction in recent years. Broadly speaking, deterministic algorithms for submodular maximization follow two main paradigms:
\textbf{(i)} designing algorithms that are deterministic by construction, and
\textbf{(ii)} derandomizing existing randomized algorithms while preserving their approximation  guarantees.

Under the first paradigm, many classical algorithms for \emph{monotone} submodular maximization are inherently deterministic.
Representative examples include greedy-like algorithms, threshold-based methods, and discrete local search algorithms.
The standard greedy algorithm achieves a $(1-1/e)$-approximation under a cardinality constraint~\cite{DBLP:journals/mp/NemhauserWF78}
and a $1/2$-approximation under matroid constraints~\cite{Fisher1978},
while Sviridenko proposed a $(1-1/e)$-approximation under a knapsack constraint~\cite{DBLP:journals/orl/Sviridenko04}.
In addition, threshold decreasing algorithms, which iteratively lower a marginal-gain threshold and select elements exceeding the current threshold,
provide a deterministic and oracle-efficient alternative to the standard greedy approach,
and achieve near-optimal approximation guarantees with significantly improved running time~\cite{DBLP:conf/soda/BadanidiyuruV14}. Buchbinder et al.~\cite{DBLP:journals/siamcomp/BuchbinderFG23} showed that
for monotone submodular maximization under a matroid constraint, the
\emph{split-and-grow} algorithm enables deterministic algorithms to
surpass the $1/2$ barrier, achieving a $(1/2+\varepsilon)$-approximation.
For \emph{non-monotone} objectives, deterministic local search algorithms were systematically studied by Feige et al.~\cite{DBLP:conf/focs/FeigeMV07},
establishing constant-factor guarantees for unconstrained submodular maximization.
Beyond these general frameworks, several algorithms have been explicitly designed to achieve strong deterministic guarantees.
Notably, the twin greedy algorithm proposed by Han et al.~\cite{DBLP:conf/nips/HanCCW20}
achieves a $1/4$-approximation for non-monotone submodular maximization under matroid constraints,
which was later extended to knapsack constraints by Sun et al.~\cite{DBLP:journals/tcs/SunZZZ24}
while preserving the same approximation ratio.
To the best of our knowledge, this remains the strongest known deterministic guarantee
for non-monotone submodular maximization under knapsack constraints.

The second paradigm focuses on \emph{derandomizing previously randomized algorithms},
and has witnessed substantial progress in recent years.
Early examples include the random greedy algorithm for cardinality constraints~\cite{DBLP:conf/soda/BuchbinderFNS14}
and the double greedy algorithm for unconstrained non-monotone maximization~\cite{DBLP:journals/siamcomp/BuchbinderFNS15},
both of which were later systematically derandomized by Buchbinder and Feldman~\cite{DBLP:conf/soda/BuchbinderF16}.
More recently, Buchbinder and Feldman~\cite{DBLP:conf/focs/BuchbinderF24}
derandomized the non-oblivious local search algorithm,
obtaining a deterministic $(1-1/e-\eps)$-approximation for monotone submodular maximization under matroid constraints.
In addition, Chen et al.~\cite{NEURIPS2024_c4e40d31}
derandomized the measured continuous greedy algorithm~\cite{DBLP:journals/mor/BuchbinderF19},
achieving approximation ratios of $0.385-\eps$ under cardinality constraints
and $0.305-\eps$ under matroid constraints.
Very recently, Buchbinder et al.~\cite{DBLP:conf/stoc/BuchbinderF25}
introduced the \emph{Extended Multilinear Extension} (EME),
a novel derandomization framework that maintains a distribution with constant-size support
and enables exact evaluation without random sampling.
Using this framework, they further obtained a deterministic $(1/e-\eps)$-approximation
for non-monotone submodular maximization under matroid constraints.



While the Extended Multilinear Extension framework offers a promising avenue
for derandomizing submodular maximization algorithms, exploiting this
framework in broader settings presents substantial technical challenges.
At a high level, the power of EME critically relies on maintaining a
distribution whose fractional support value remains constant throughout the
algorithm. This structural requirement fundamentally alters the nature
of the optimization process: unlike classical multilinear extension
methods, the optimization under EME can no longer be carried out via
purely continuous procedures.
Instead, each update step must carefully combine continuous adjustments
with embedded combinatorial subroutines that preserve the bounded support
property. Designing such hybrid steps is  nontrivial, as standard
continuous optimization techniques are no longer applicable. Moreover,
the existing derandomization approach for EME is tightly coupled with the
exchange properties of matroids, and the resulting algorithmic framework
relies heavily on matroid-specific combinatorial structures. As a
consequence, these techniques do not readily extend to other algorithms and constraint
families.
For more general combinatorial constraints, additional difficulties
arise: one must design new rounding procedures that are compatible with
the EME framework while simultaneously guaranteeing that the support size
remains constant. Balancing these requirements poses a central obstacle
to extending EME-based derandomization beyond the matroid setting.  These
challenges naturally raise the following question:

\emph{Can we leverage the EME framework to derandomize continuous optimization algorithms and achieve improved approximation ratios under different combinatorial constraints?}

\paragraph{Our Contribution.}
In our work, we consider the problem of maximizing a non-monotone submodular function subject to constraints of the form
\[
\max \{ f(S) : S \in \mathcal{C} \}
\]
where $\mathcal{C}$ denotes the constraint of the submodular maximization problem, that is, the collection of feasible sets:
\begin{itemize}
\item \textbf{Matroid constraints.} 
A \emph{matroid} $\mathcal{M} = (\mathcal{N}, \mathcal{I})$ is a set system, where $\mathcal{I} \subseteq 2^{\mathcal{N}}$ is a family of subsets of the ground set $\mathcal{N}$ satisfying the following properties: (i) $\emptyset \in \mathcal{I}$; (ii) If $A \subseteq B$ and $B \in \mathcal{I}$, then $A \in \mathcal{I}$; (iii) If $A, B \in \mathcal{I}$ with $|A| < |B|$, there exists an element $u \in B \setminus A$ such that $A + u \in \mathcal{I}$. The members of $\mathcal{I}$ are called \emph{independent sets}, and the maximal independent sets are called \emph{bases}. All bases of a matroid have the same cardinality, which is referred to as the \emph{rank} of $\mathcal{M}$. For a subset $S \subseteq \mathcal{N}$ (not necessarily independent), its \emph{rank} is defined as the size of its maximum independent subset, i.e.,
\(
\Rank(S) = \max_{T \subseteq S,\; T \in \mathcal{I}} |T|.
\)
The \emph{span} of a set $S \subseteq \mathcal{N}$ is defined as the set of elements $u \in \mathcal{N}$ such that $r=\Rank(S \cup \{u\}) = \Rank(S)$.

\item \textbf{Knapsack constraints.} Each element $u \in \mathcal{N}$ is associated with a non-negative cost $w(u)$, and a set $S \subseteq \mathcal{N}$ is feasible if
    $
    \sum_{u \in S} w(u) \le B,
    $
    where $B > 0$ denotes the knapsack budget. For simplicity, We denote the weight of a set $S$ by $w(S) =  \sum_{u \in S} w(u)$.
\end{itemize}

We firstly address these challenges by derandomizing the aided continuous greedy algorithm under matroid constraints. Our approach exploits a discrete stationary point to aid the optimization trajectory and performs a piecewise optimization over the EME, thereby maintaining a constant-size support throughout the process.
However, extending the convergence analysis to the EME setting is nontrivial. Existing convergence analyses of the aid measured continuous greedy algorithm rely on the Lovász extension as a lower bound of the multilinear extension ~\cite{DBLP:journals/mor/BuchbinderF19} or on the directional concavity of the multilinear extension of a submodular function~\cite{DBLP:journals/corr/abs-2307-09616}. These properties do not carry over to the EME, which implies us adopt a new analysis process. When combined with the rounding scheme of Buchbinder et al.~\cite{DBLP:conf/stoc/BuchbinderF25}, our framework yields a deterministic $(0.385-\varepsilon)$-approximation. This strictly improves upon the previously best-known deterministic guarantee of $(1/e-\varepsilon)$ under matroid constraints, which was also achieved in \cite{DBLP:conf/stoc/BuchbinderF25}. The formal statement is as follows:

\begin{theorem}
    \label{matroid_main}
    \label{mainthm}
Given a non-negative submodular function $f:2^\NN\rightarrow \R_{\ge 0}$ and a matroid $\mathcal{M} = (\mathcal{N}, \mathcal{I})$, Algorithm~\ref{alg:matroid_main} is a deterministic algorithm that uses {$O_{\eps}(n^5)$} queries,
returns a set $S$ satisfying the matroid constraint $\mathcal{M} = (\mathcal{N}, \mathcal{I})$, and achieves $f(S)\ge (0.385-O(\eps))f(\OPT)$.
\end{theorem}

For knapsack constraints, the lack of matroid exchange properties
necessitates a different and more delicate approach. We begin by
enumerating a set of $O(1/\varepsilon^{2})$ elements, which allows us to
assume without loss of generality that all remaining elements are
\emph{small} enough. Under this small-element assumption, we develop an
optimization algorithm over the EME that is guided by element densities.
This density-based optimization admits a provable approximation guarantee
on the EME solution and can be implemented while maintaining a
constant-size support.

We then apply a rounding procedure inspired by Pipage Rounding, which
repeatedly selects two non-integral elements in the support of the EME
and moves along a suitable convex direction to integralize one of them. Unlike the matroid setting,
this process may leave at most one element in a fractional state.
However, since this element is guaranteed to be small, feasibility can be
ensured by executing the optimization step with a reduced knapsack
capacity of $(1-\varepsilon)$, thereby reserving sufficient slack to
accommodate the final rounding. Overall, this framework yields a
deterministic $(1/e-\varepsilon)$-approximation for submodular
maximization under a knapsack constraint, substantially improving upon the $1/4$ approximation ratio of twin greedy  methods.




\begin{theorem}
    \label{knapsack_main}
Given a non-negative submodular function $f:2^\NN\rightarrow \R_{\ge 0}$ and a weight function $w:\NN \rightarrow \R_{\ge 0}$, Algorithm ~\ref{Knapsack:whole} is a deterministic algorithm that uses {$O_{\eps}(n^{\eps^{-2}})$} queries, returns a set $S$ satisfying the knapsack constraint $\sum_{u\in S}w(u)\le B$, and achieves $f(S)\ge (1/e-O(\eps))f(\OPT)$.
\end{theorem}

\begin{table}[]
\centering
\caption{Comparison of algorithms for non-monotone submodular maximization under matroid and knapsack constraints.
The table includes state-of-the-art results, our deterministic algorithms, and  algorithms we derandomize (marked with $\dagger$).}

\begin{tabular}{lllll}

\Xhline{1.5pt}
Constraints                & Approximation Ratio & Query Complexity          & Type                            & Reference                   \\ \hline
                           & $0.385 - \eps$      & $O_{\eps}(n^{11})$          &                                 & \cite{DBLP:journals/mor/BuchbinderF19}{$^\dagger$}                     \\ 
                           & $0.401 - \eps$      & $\mathrm{poly}(n,1/\eps)$           & \multirow{-2}{*}{Random}        & \cite{DBLP:conf/stoc/BuchbinderF24}                          \\ \cline{2-5}
                           & $0.305 -\eps$             & $O_{\eps}(nk)$            &                                 & \cite{NEURIPS2024_c4e40d31}                           \\ 
                           & $0.367 - \eps$      & $O_{\eps}(n^5)$           &                                 & \cite{DBLP:conf/stoc/BuchbinderF25}                           \\ 
\multirow{-5}{*}{Matroid}  & $0.385 - \eps$      & $O_{\eps}(n^5)$           & \multirow{-3}{*}{Deterministic} & {\textbf{Ours}} \\ \hline
                           & $0.367 - \eps$      & $O_{\eps}(n^{\eps^{-2}})$ &                         & \cite{DBLP:conf/focs/FeldmanNS11}{$^\dagger$}                          \\ 
                           & $0.401 - \eps$      & $O_{\eps}(n^{\eps^{-2}})$ & \multirow{-2}{*}{Random}                                & \cite{DBLP:conf/stoc/BuchbinderF24}                           \\ \cline{2-5}
                           & $0.25$              & $O(n^4)$                    &                                 & \cite{DBLP:journals/tcs/SunZZZ24}                           \\
\multirow{-4}{*}{Knapsack} & $0.367 - \eps$      & $O_{\eps}(n^{\eps^{-2}})$ & \multirow{-2}{*}{Deterministic} & {\textbf{Ours}} \\ \Xhline{1.5pt}
\end{tabular}
\end{table}
\paragraph{Organization.}
Our paper is organized as follows: in Section~2, we introduce basic notations and review the definition and fundamental properties of the extended multilinear extension. In Section~3, we present our algorithm for the matroid constraint: Section~3.1 describes the discrete local search procedure, Section~3.2 explains how a local search stationary point guides the Deterministic aided Continuous Greedy algorithm, and Section~3.3 analyzes the performance of the overall algorithm. In Section~4, we present our algorithm for the knapsack constraint: Section~4.1 begins with the enumeration of elements, Section~4.2 introduces a deterministic continuous greedy algorithm for the knapsack constraint, Section~4.3 presents our rounding procedure, and Section~4.4 analyzes the performance of the entire algorithm. Finally, in Section~5, we summarize our work and discuss possible directions for future research.
\section{Preliminary}
\label{sec:pre}
In this section, we first introduce some basic definitions and essential lemmas adopted in our algorithm design. Without  ambiguity, we denote $n=|\N|$ and use $\OPT$ to represent the optimal feasible solution. We do not explicitly consider how the submodular function $f$ is computed; instead, we assume access to a \emph{value oracle} that returns $f(S)$ for any set $S \subseteq \mathcal{N}$ in $O(1)$ time. In addition, for problems under matroid constraints, we assume access to a \emph{membership (or independence) oracle} that determines whether a given set is independent in $O(1)$ time. The total number of calls to the value oracle and the membership oracle made by an algorithm is referred to as its \emph{query complexity}.
\paragraph{Polytope of Constraints.}
Algorithms based on the continuous optimization tool are typically executed over a feasible region, whose form depends on the underlying constraints and is referred to as the \emph{polytope} associated with the corresponding constraint.
The matroid polytope $\mathcal{P}(\mathcal{M}) \subseteq [0, 1]^{\N}$ is defined as:
\begin{align*}
    \mathcal{P}(\mathcal{M}) = \text{conv}\left\{ \mathbf{1}_{S} \mid S \in \mathcal{I} \right\}
\end{align*}
where $\mathbf{1}_{S} \in [0, 1]^{\N}$ is a vector such that for any $u \in S$, $(\mathbf{1}_{S})_u = 1$, and for any $u \in \N \setminus S$, $(\mathbf{1}_{S})_u = 0$. The knapsack polytope is defined as 
\[ 
\mathcal{P}(B) = \left\{  x \in [0,1]^{\N}:\sum_{u \in \N}  w(u) \cdot  x_u \le B \right\}
 \]
\paragraph{Extended Multilinear Extension.}
The paper from \cite{DBLP:conf/stoc/BuchbinderF25} proposed an extension of the multilinear extension to address the issue that the standard multilinear extension cannot yield deterministic algorithms. For a submodular function $f: 2^{\mathcal{N}} \rightarrow \mathbb{R}$, its extended multilinear extension $F: [0, 1]^{2^{\mathcal{N}}} \rightarrow \mathbb{R}$ is defined as:
\begin{align*}
    F(\symy) = \sum_{\mathcal{J} \subseteq 2^{\mathcal{N}}} \left( f\left( \bigcup_{S \in \mathcal{J}} S \right) \cdot \prod_{S \in \mathcal{J}} \symy_S \cdot \prod_{S \notin \mathcal{J}} (1 - \symy_S) \right)
\end{align*}

Below we introduce some notation and properties related to the extended multilinear extension. All the  properties can be found in \cite{DBLP:conf/stoc/BuchbinderF25}.

\begin{definition}[Random set]
    For a vector $\symy \in [0, 1]^{2^{\mathcal{N}}}$, define the random set $\reali(\symy, S)$, where $\Pr[\reali(\symy, S) = S] = \symy_S$, otherwise $\reali(\symy, S) = \emptyset$. Furthermore, define $\reali(\symy) = \bigcup_{S \subseteq \mathcal{N}} \reali(\symy, S)$. 
\end{definition}
\begin{observation}
    For any $\symy\in[0,1]^{2^{\N}}$ $F(\symy) = \mathbb{E} [f(\reali(\symy))]$.
\end{observation}

\begin{definition}[Coordinate-wise probabilistic sum] The coordinate-wise probabilistic sum of $m$ vectors $\mathbf{a}^{(1)}, \mathbf{a}^{(2)}, \dots, \mathbf{a}^{(m)} \in [0, 1]^n$ is a vector $\mathbf{c}$ where each entry $k \in \{1, \dots, n\}$ is defined as: \[ \left( \bigoplus_{i=1}^{m} \mathbf{a}^{(i)} \right)_k = 1 - \prod_{i=1}^{m} (1 - a_k^{(i)}) \]
    
\end{definition}

\begin{definition}[Marginal vector of $\mathbf{y}$]
Let $\mathbf{y} \in [0, 1]^{2^{\mathcal{N}}}$, then the marginal vector $\text{Mar}(\mathbf{y}) \in [0, 1]^{\mathcal{N}}$ is defined by:\[ \text{Mar}_u(\mathbf{y}) \triangleq \Pr[u \in R(\mathbf{y})] = 1 - \prod_{S \subseteq 2^{\mathcal{N}} \mid u \in S} (1 - y_S) \quad \forall u \in \mathcal{N} \]
\end{definition}

\begin{observation}
\label{maroplus}
Let $\mathbf{y}^1, \mathbf{y}^2 \in [0, 1]^{2^{\mathcal{N}}}$ be vectors such that $\mathbf{y}^1 + \mathbf{y}^2 \le \mathbf{1}$. Then:\[ \text{Mar}(\mathbf{y}^1 \oplus \mathbf{y}^2) = \text{Mar}(\mathbf{y}^1) \oplus \text{Mar}(\mathbf{y}^2) \]
\end{observation}

\begin{observation}
\label{observation 8}
 Given a set function $f: 2^{\mathcal{N}} \to \mathbb{R}$ and a vector $\mathbf{y} \in [0,1]^{2^{\mathcal{N}}}$, we define the set function $g_{\mathbf{y}}: 2^{\mathcal{N}} \to \mathbb{R}$ by
\begin{equation}
    g_{\mathbf{y}}(A) \triangleq F(\mathbf{e}_A \vee \mathbf{y}) \quad \forall A \subseteq \mathcal{N}.
\end{equation}
Where $e_A\in [0,1]^{2^\N}$ is a vector such that $(e_{A})_A=1$ and for any $ S\subseteq \N, S\neq A $, $(e_A)_S=0$.
If $f$ is a non-negative submodular function, then $g_{\mathbf{y}}(A)$ is also a non-negative submodular function. 
\end{observation}

\begin{observation}
\label{obs:eme}
 Let $F$ be the extended multilinear extension of an arbitrary set function $f: 2^{\mathcal{N}} \to \mathbb{R}$. Then, for every vector $\mathbf{y} \in [0, 1]^{2^{\mathcal{N}}}$,

\begin{align*}
F(\mathbf{e}_S \vee \mathbf{y}) &\geq (1 - \|\operatorname{Mar}(\mathbf{y})\|_\infty) \cdot f(S) &&\forall  S \subseteq \mathcal{N} \\
    \frac{\partial^2 F}{(\partial y_S)^2}(\mathbf{y}) &= 0 && \forall S \subseteq \mathcal{N} \nonumber \\
        \frac{\partial^2 F}{\partial y_S \partial y_T}(\mathbf{y}) &\leq 0 \qquad &&\forall S, T \subseteq \mathcal{N}, S \cap T = \varnothing\\
    (1 - y_S) \cdot \frac{\partial F(\mathbf{y})}{\partial y_S} &= F(\mathbf{e}_S \vee \mathbf{y}) - F(\mathbf{y}) && \forall S \subseteq \mathcal{N}, y_S \in [0, 1] \\
    F(\mathbf{y} \oplus \mathbf{z}) &= \sum_{J \subseteq 2^{\mathcal{N}}} \left( F\left( \mathbf{y} \vee \sum_{S \in J} \mathbf{e}_S \right) \cdot \prod_{S \in J} z_S \cdot \prod_{S \in 2^{\mathcal{N}} \setminus J} (1 - z_S) \right) && \forall \mathbf{z} \in [0, 1]^{2^{\mathcal{N}}} 
\end{align*}
    
\end{observation}

\paragraph{\lovasz Extension.}
The \lovasz extension is another important extension for set functions. However, in this paper, it is only used as a definition for auxiliary proof conclusions, so its properties will not be extensively covered. For a submodular function $f: 2^{\mathcal{N}} \rightarrow \mathbb{R}$, its \lovasz{} extension $\hat{f}: [0, 1]^{\mathcal{N}} \rightarrow \mathbb{R}$ is defined as:
\begin{align*}
    \hat{f}(\symx) = \int_{0}^{1} f(T_{\lambda}(\symx))  d\lambda
\end{align*}
where $T_\lambda(\symx) = \{ u \in \mathcal{N} : \symx_u > \lambda \}$.

\begin{lemma}
    \label{randomset}
    Given $\symx \in [0, 1]^{\mathcal{N}}$ and a submodular function $f: 2^{\mathcal{N}} \rightarrow \mathbb{R}$. Let $\hat{f}: [0, 1]^{\mathcal{N}} \rightarrow \mathbb{R}$ be its \lovasz{} extension. Let $A(\symx) \subseteq \mathcal{N}$ be a random set satisfying $\forall u\in \N,\Pr[u \in A(\symx)] = \symx_u$. Then:
    \begin{align*}
        \mathbb{E} [f(A(\symx))] \ge \hat{f}(\symx)
    \end{align*}
\end{lemma}
\begin{proof}
    Let $\{u_1, u_2, \dots, u_n\} = \mathcal{N}$ be ordered such that $\forall i < j$, $\symx_{u_i} \ge \symx_{u_j}$. Let event $X_i$ indicate $u_i \in A(\symx)$. Let $A_i = \{u_j \mid j \le i\}$.
    \begin{align*}
        \mathbb{E} [f(A(\symx))] &= \mathbb{E} \left[ f(\emptyset) + \sum_{i=1}^n f(A(\symx) \cap A_i \mid A(\symx) \cap A_{i-1}) \right] \\
        &= \mathbb{E} \left[ f(\emptyset) + \sum_{i=1}^n X_i \cdot f(u_i \mid A_{i-1} \cap A(\symx)) \right] \\
        &\ge \mathbb{E} \left[ f(\emptyset) + \sum_{i=1}^n X_i \cdot f(u_i \mid A_{i-1} ) \right] = f(\emptyset) + \sum_{i=1}^n \mathbb{E}[X_i] \cdot  f(u_i \mid A_{i-1} )  \\
        &= f(\emptyset) + \sum_{i=1}^n \symx_{u_i} \cdot f(u_i \mid A_{i-1}) = f(\emptyset) + \int_{0}^{1} \sum_{\substack{i \\ \lambda < \symx_{u_i}}} f(u_i \mid A_{i-1})  d\lambda \\
        &= \hat{f}(\symx)
    \end{align*}
    The first equality decomposes the function value into the incremental contribution of each element. The inequality follows from submodularity. The final equality holds because, due to the descending order of $\symx_{u_i}$, the sum (including the preceding empty set term) constitutes a difference sequence for $f(T_\lambda(\symx))$, thus equaling $\hat{f}(\symx)$.
\end{proof}

\section{Matroid Constraints}
In this section, we introduce the deterministic aided measured continuous greedy algorithm based on the extended multilinear extension. In \cite{DBLP:conf/stoc/BuchbinderF25}, the authors have provided a rounding method for EME vectors when the support set is sufficiently small. Consequently, utilizing this rounding algorithm, our algorithm only needs to obtain a vector with a sufficiently small support set to provide a solution for maximizing submodular functions. Additionally, we introduce a dummy element set $D$, which only occupy constraints without changing the function value, and the definition of the discrete stationary point, which are essential bridges in our theoretical analysis.                \paragraph{Deterministic-Pipage}
In \cite{DBLP:conf/stoc/BuchbinderF25}, the authors provide a deterministic pipage rounding procedure that converts an EME vector with sufficiently small support into a feasible set in the matroid. Therefore, our algorithm only needs to compute a vector whose support size is sufficiently small in order to obtain a solution for the submodular maximization problem.

\begin{theorem}[Deterministic-Pipage, Theorem 5.6 in \cite{DBLP:conf/stoc/BuchbinderF25}]
    \label{rounding}
    For a vector $y \in [0,1]^{2^{\mathcal{N}}}$, if $\mar(y) \in \mathcal{P}(\mathcal{M})$, there exists an algorithm that, within $O(n^2 \cdot 2^{\ffrac(y)})$ value queries and $O(n^5 \log^2 n)$ independence queries, returns a set $S \in \mathcal{M}$ satisfying $f(S) \ge F(y)$.
\end{theorem}

\paragraph{Dummy elements.}
We can make the subsequent analysis more convenient by adding at most $n$ dummy elements to the original set $\mathcal{N}$. Specifically, let $\bar{\mathcal{N}} = \{u_1, \dots, u_n\}$, $\bar{\mathcal{I}} = \{ S \subseteq \mathcal{N} \cup \bar{\mathcal{N}} \mid S \cap \mathcal{N} \in \mathcal{I}, |S| \le \rank(\mathcal{M}) \}$, and $\bar{f}: 2^{\mathcal{N} \cup \bar{\mathcal{N}}} \rightarrow \mathbb{R}$, $\bar{f}(S) = f(S \setminus \bar{\mathcal{N}})$. The solution to the new problem $(\mathcal{N} \cup \bar{\mathcal{N}}, \bar{\mathcal{I}}), \bar{f}$ will be consistent with the original problem. Additionally, we can always add dummy elements to a set $S \in \mathcal{I}$ to make $S$ a basis of the matroid $\mathcal{M}$ without changing the function value of $S$. The computed result will inevitably contain some dummy elements; we simply remove them from the solution set to obtain a solution to the original problem with the same function value. In subsequent analysis, we will assume by default that the problem includes a sufficient number of dummy elements.

\begin{lemma}
    $\bar{\mathcal{M}} = (\mathcal{N} \cup \bar{\mathcal{N}}, \bar{\mathcal{I}})$ is a matroid, and $\bar{f}$ is a submodular function.
\end{lemma}
\begin{proof}
    For any $A \subseteq B \subseteq \mathcal{N} \cup \bar{\mathcal{N}}$, $u \in (\mathcal{N} \cup \bar{\mathcal{N}}) \setminus B$:
    If $u \in \bar{\mathcal{N}}$, then $\bar{f}(A + u) - \bar{f}(A) = \bar{f}(B + u) - \bar{f}(B) = 0$. Otherwise,
    \begin{align*}
        \bar{f}(A + u) - \bar{f}(A) &= f((A \setminus \bar{\mathcal{N}}) + u) - f(A \setminus \bar{\mathcal{N}}) \\
        &\ge f((B \setminus \bar{\mathcal{N}}) + u) - f(B \setminus \bar{\mathcal{N}}) \\
        &= \bar{f}(B + u) - \bar{f}(B)
    \end{align*}
    Therefore, $\bar{f}$ is submodular.
    
    To prove $\bar{\mathcal{M}}$ is a matroid, we only need to prove $\bar{\mathcal{I}}$ satisfies the exchange property.
    For any $A, B \in \bar{\mathcal{I}}$ with $|A| < |B|$. If $|A \cap \bar{\mathcal{N}}| < |B \cap \bar{\mathcal{N}}|$, then there exists a dummy element $u \in (B \cap \bar{\mathcal{N}}) \setminus A$, so $A + u \in \bar{\mathcal{I}}$.
    Otherwise, consider $|A \setminus \bar{\mathcal{N}}| = |A| - |A \cap \bar{\mathcal{N}}| < |B| - |B \cap \bar{\mathcal{N}}| = |B \setminus \bar{\mathcal{N}}|$. By the definition of $\bar{\mathcal{I}}$, $A \in \bar{\mathcal{I}} \Rightarrow A \setminus \bar{\mathcal{N}} \in \mathcal{I}$,
    and $A \setminus \bar{\mathcal{N}}, B \setminus \bar{\mathcal{N}} \in \mathcal{I}$. By the matroid exchange property, there exists $u \in B \setminus \bar{\mathcal{N}} \subseteq B$ such that $(A \setminus \bar{\mathcal{N}}) + u \in \mathcal{I}$, implying $A + u \in \bar{\mathcal{I}}$.
\end{proof}

The main idea of the algorithm on the extended multilinear extension is similar with that on the multilinear extension. First, the main algorithm finds a reference solution $\symz$ through the local search algorithm. Then, adopting the information from $\symz$, the main algorithm performs continuous greedy in a direction orthogonal to $\symz$ before time $t_s$ and perform continuous greedy in all direction after time $t_s$, obtaining a fractional solution $\mathbf{y}$ negatively correlated with $f(\symz)$. Finally, the best of $Rounding(\mathbf{y})$ and $\symz$ is an $0.385- O(\eps)$ approximation solution. The formal statement of the deterministic $0.385-O(\eps)$ approximation algorithm for non-monotone submodular maximization subject to matroid constraints is as follows:


\begin{algorithm}[H]        
    \label{alg:matroid_main}
  \caption{MAIN$(\mathcal{M}=(\mathcal{N},\mathcal{I}), f, t_s \in [0,1], \eps \in (0,1))$}
  \label{alg4}
  \small
  $\symz \gets \text{\textsc{LocalSearch}}(\mathcal{M}, f, \eps)$\;
  $\symy^1 \gets \text{\textsc{AidedContinuousGreedy}}(\mathcal{M}, f, \symz, t_s, \eps)$\;
  \Return $\arg\max\limits_{Y \in \{\textsc{Deterministic-Pipage}(\symy^1), \symz\}} f(Y)$\;
\end{algorithm}

\subsection{Discrete Local Search}
\label{localsearch}
In this subsection, we introduce a discrete local search algorithm to obtain the reference solution $\symz$. Although this algorithm originates from Algorithm 3 in \cite{DBLP:conf/focs/BuchbinderF24}, we restate and reprove it such that it can be adapted to the non-monotone submodular maximization problem we study. The algorithm calls a deterministic combinatorial algorithm to obtain a constant-factor approximation of $f(OPT)$ as its initial solution and continuously swaps elements to ensure that the current solution is a basis of the matroid $\mathcal{M}=(\N,\I)$ and the objective function value increases. Essentially, the solution returned by Algorithm \ref{alg1} is a discrete stationary point. The formal statement is as follows:

\begin{algorithm}[H]
  \caption{\textsc{LocalSearch}$(\mathcal{M}=(\mathcal{N},\mathcal{I}), f, \eps \in (0,1))$}
  \label{alg1}
  \small
  Use the algorithm in Lemma~\ref{approx} to find $S_0$ such that $f(S_0) \ge 0.305 \cdot f(OPT)$, set $S \gets S_0$\;
  Add elements until $S$ becomes a basis of $\mathcal{M}$\;
  \While{\rm there exist $u \in S$, $v \in \mathcal{N} \setminus S$ such that $S - u + v \in \mathcal{I}$ and $f(v \mid S) - f(u \mid S - u) \ge \frac{\eps}{r} \cdot f(S_0)$}{
    $S \gets S + v - u$\;
  }
  \Return $S$\;
\end{algorithm}

\begin{lemma}[Theorem 2.4 in \cite{NEURIPS2024_c4e40d31}]
    \label{approx}
    For a non-monotone submodular maximization problem $(\mathcal{M}, f)$, there exists a deterministic algorithm that returns an $0.305 - \eps$ approximate solution with at most $O_{\eps}(n r)$ value and independence queries.
\end{lemma}
\paragraph{Discrete Stationary Points.}
Intuitively, we call a solution stationary if no local exchange can 
significantly improve the objective value, meaning the algorithm has 
reached a stable state with respect to $1$-exchange moves.

We formalize the notion of a discrete stationary point under $1$-exchange.
Let $\mathcal{M}=(\mathcal{N},\mathcal{I})$ be a matroid of rank $r$ and 
$f:2^{\mathcal{N}}\to\mathbb{R}$ be submodular. 
A basis $S \in \mathcal{I}$ is called an $\alpha$-approximate 
discrete stationary point if for every pair $u \in S$ and 
$v \in \mathcal{N}\setminus S$ such that $S-u+v \in \mathcal{I}$,
\[
f(v \mid S) - f(u \mid S-u) \le \alpha .
\]
When $\alpha=0$, this coincides with the standard notion of a 
$1$-exchange local optimum.

Algorithm~\ref{alg1} terminates exactly when no feasible exchange 
improves the objective by more than 
$\frac{\varepsilon}{r}\cdot f(S_0)$; hence the returned solution 
is an $\alpha$-approximate discrete stationary point with 
$\alpha = \frac{\varepsilon}{r}\cdot f(S_0)$.

For notational simplicity, throughout the sequel we refer to such 
$\alpha$-approximate discrete stationary points simply as 
\emph{stationary points}, as the additive error will be explicitly 
accounted for in the analysis. Next, we illustrate the theoretical guarantee for Algorithm \ref{alg1} as follows:
\begin{theorem}
    \label{thm2.1}
    Algorithm \ref{alg1} terminates within $O(\eps^{-1} n r)$ value queries and $O(\eps^{-1} n r^2)$ independence queries, and the returned set $S$ satisfies $S \in \mathcal{I}$ and, for any $T \in \mathcal{I}, \eps \in (0,1)$,
    \begin{align}
        \label{eq1}
        f(S) \ge \frac{1}{2} (f(S \cap T) + f(S \cup T)) - \eps \cdot f(OPT)
    \end{align}
\end{theorem}

\begin{proof}
    Due to dummy elements, we can assume $|S| = |T| = r$, both being bases of $\mathcal{M}$. Let $\{u_1, \dots, u_k\}$ and $\{v_1, \dots, v_k\}$ denote the elements in $S \setminus T$ and $T \setminus S$, respectively. By the matroid exchange property, we can find a bijection $h: S \setminus T \rightarrow T \setminus S$ such that $S - u_i + h(u_i) \in \mathcal{I}$ for each $u_i$.
    \begin{align*}
        f(S \cup T) + f(S \cap T) - 2f(S) 
        = &  (f(S \cup T) - f(S)) - (f(S) - f(S \cap T)) \\
        \le &\left( \sum_{i=1}^k f(v_i \mid S) \right) - \left( \sum_{i=1}^k f(u_i \mid S - u_i) \right) \quad \text{(by submodularity)} \\
        = &\sum_{i=1}^k \left[ f(h(u_i) \mid S) - f(u_i \mid S - u_i) \right] \\
        \le &k \cdot \frac{\eps}{r} \cdot f(S_0) \le \eps \cdot f(OPT) \quad \text{(by loop condition)}
    \end{align*}
    Consequently, the solution returned by Algorithm \ref{alg1} satisfies the inequality (\ref{eq1}). In each iteration, we have:
    \begin{align*}
        f(S + v - u) - f(S) &= (f(S + v - u) - f(S - u)) + (f(S - u) - f(S)) \\
        &\ge (f(S + v) - f(S)) + (f(S - u) - f(S)) \quad \text{(by submodularity)} \\
        &= f(v \mid S) - f(u \mid S - u) \\
        &\ge \frac{\eps}{r} \cdot f(S_0) > \frac{\eps}{4r} \cdot f(OPT)
    \end{align*}
    The inequality in the algorithm's loop condition ensures the last step. Since $f(S) - f(S_0) \le f(OPT) - f(S_0) \le f(OPT)$, the loop must terminate after at most $\frac{4r}{\eps}$ iterations; otherwise, it would contradict $f(S) - f(S_0) \le f(OPT)$. In each iteration, the algorithm requires $O(n)$ value oracle queries and $O(nr)$ independence oracle queries. Totally, Algorithm \ref{alg1} needs $O(\eps^{-1}nr)$ value oracle queries and $O(\eps^{-1}nr^2)$ independence oracle queries.
\end{proof}

\subsection{Deterministic aided Continuous Greedy}
In this subsection, we introduce the aided continuous greedy algorithm that leverages the discrete stationary point returned by the local search algorithm to guide the forward direction on continuous greedy algorithm. Firstly, we introduce a combinatorial algorithm for obtaining a forward direction. The continuous greedy process then proceeds by taking a small step along this direction.
By construction, the direction obtained at each step yields a sufficiently large marginal gain, and thus the algorithm advances by updating the current solution by a small constant multiple of this vector.

\paragraph{Split Algorithm.}
\label{split}
The split algorithm serves to extract a forward direction with the largest marginal gain at the current point, under the constraint that the support size remains bounded.
This constraint is crucial from a computational perspective: evaluating the extended multilinear extension value incurs a cost that scales with the size of its support. The split algorithm returns a partition of an independent set. As the number of partitions are bounded, the size of support set for the forward direction is also bounded by $\ell$. Although the theoretical guarantees for Algorithm \ref{alg2} was proposed in \cite{DBLP:conf/stoc/BuchbinderF25}, for completeness, we also restate the proof based on the \lovasz extension, which has some subtle differences from \cite{DBLP:conf/stoc/BuchbinderF25}. 

\begin{algorithm}[H]
  \caption{\textsc{Split}$(\mathcal{M}=(\mathcal{N},\mathcal{I}), f, \ell)$}
  \label{alg2}
  \small
  Initialize $T_1 \gets \emptyset,\; T_2 \gets \emptyset,\; \dots,\; T_\ell \gets \emptyset$\;
  Let $T \gets \bigcup_{i=1}^\ell T_i$\;
  \While{$T$ is not a basis of $\mathcal{M}$}{
    $\mathcal{N}' \gets \{ u \in \mathcal{N} \setminus T \mid T + u \in \mathcal{I} \}$\;
    $(u, j) \gets \arg\max_{(u, j) \in \mathcal{N}' \times [\ell]} f(u \mid T_j)$\;
    $T_j \gets T_j + u$\;
  }
  \Return $(T_1, T_2, \dots, T_\ell)$\;
\end{algorithm}
\begin{lemma}[Lemma 4.2 in \cite{DBLP:conf/stoc/BuchbinderF25}]
    \label{lemmaspliteq}
    Given a matroid $\mathcal{M}$ on the ground set $\mathcal{N}$, a submodular function $f$, and an integer $\ell$ as input, Algorithm \ref{alg2} outputs disjoint sets $(T_1, T_2, \dots, T_\ell)$ whose union is a basis of $\mathcal{M}$, satisfying for any set $O \subseteq \mathcal{N}$:
    \begin{align}
    \label{spliteq}
        \sum_{j=1}^{\ell} f(T_j \mid \emptyset) \ge \left(1 - \frac{1}{\ell}\right) \cdot f(O) - \frac{1}{\ell} \sum_{j \in [\ell]} f(T_j)
    \end{align}
\end{lemma}
\begin{proof}
    Let $T = \bigcup_{j=1}^\ell T_j$. Due to dummy elements, we can assume $|T| = |O| = r$.
    By the exchange property of matroid, there exists a bijection $h: O \rightarrow T$ such that for any $u \in O \setminus T$, $(T - h(u)) + u \in \mathcal{I}$. Furthermore, for $u \in O \cap T$, $h(u) = u$.
    Let $u_i \in T \setminus O$ be the $i$-th element added to $T$ in Algorithm \ref{alg2}. Define $T_j^i = T_j \setminus \{ h^{-1}(u_k) \mid k > i \}$, and let $j_i$ denote the index of the set to which $u_i$ belongs, i.e., $u_i \in T_{j_i}$. According to the algorithm, we state:
    \begin{align*}
        f(u_i \mid T_{j_i}^{i-1}) \ge \frac{1}{\ell} \sum_{j=1}^\ell f(h^{-1}(u_i) \mid T_{j}^{i-1}) \ge \frac{1}{\ell} \sum_{j=1}^\ell f(h^{-1}(u_i) \mid T_{j})
    \end{align*}
    The first inequality holds because the exchange property guarantees $(T - u_i) + h^{-1}(u_i) \in \mathcal{I}$, so the greedy strategy of the algorithm ensures $f(u_i \mid T_{j_i}^{i-1}) \ge f(h^{-1}(u_i) \mid T_{j}^{i-1})$ for any $j$ (otherwise $h^{-1}(u_i)$ would have been added to some $T_j$ at this step). The second inequality follows from the submodularity of $f$.
    Summing the above equation over all $i$ gives:
    \begin{align*}
        \sum_{j=1}^\ell f(T_j \mid \emptyset) &= \sum_{i=1}^{r} f(u_i \mid T_{j_i}^{i-1}) \nonumber \ge \frac{1}{\ell} \sum_{j=1}^\ell \sum_{i=1}^r f(h^{-1}(u_i) \mid T_j) \nonumber \\
        &\ge \frac{1}{\ell} \sum_{j=1}^\ell f\left( \bigcup_{i=1}^r \{ h^{-1}(u_i) \} \mid T_j \right) \quad \text{(by submodularity)} \nonumber \\
        &= \frac{1}{\ell} \sum_{j=1}^\ell \left( f(T_j \cup O) - f(T_j) \right) 
    \end{align*}
    Let $A$ be a random set where for each $1 \le j \le \ell$, $\Pr[A = T_j \cup O] = \frac{1}{\ell}$. Then $\frac{1}{\ell} \sum_{j=1}^{\ell} f(T_j \cup O) = \mathbb{E} [f(A)]$.
    Let $p_u = \Pr[u \in A]$. Since the $T_j$ are disjoint, $p_u \le \frac{1}{\ell}$ for all $u \in \mathcal{N} \setminus O$, and $p_u = 1$ for all $u \in O$.
    By Lemma \ref{randomset},
    \begin{align*}
        \mathbb{E} [f(A)] \ge \hat{f}(p) \ge \int_{\frac{1}{\ell}}^1 f(T_{\lambda}(p))  d\lambda \ge \left(1 - \frac{1}{\ell}\right) f(O)
    \end{align*}
    The first inequality is from Lemma \ref{randomset}, the second from the definition of the Lovász extension, and the third follows from the property of $p$: only elements in $O$ have probability $> \frac{1}{\ell}$ of being in the random set.
    Combining these inequalities yields (\ref{spliteq}).
\end{proof}
Corollary~\ref{splitprop} specializes Lemma~\ref{lemmaspliteq} to the regime $\ell = 1/\varepsilon$ and rewrites the inequality in terms of total marginal gain with respect to the empty set. As we will show later, this form allows us to establish a lower bound on the improvement obtained by the forward direction induced by the split algorithm. Essentially, Corollary~\ref{splitprop} provides a quantitative guarantee on the quality of the direction induced by the split algorithm. When the function $f$ is instantiated as the objective evaluated at the current fractional solution, the left-hand side represents the total marginal gain contributed by the components of the forward direction. The inequality then implies that this direction achieves a sufficiently large increase relative to any independent set $O$, and therefore constitutes a sufficiently good forward direction for the continuous greedy process.
\begin{corollary}
    \label{splitprop}
    In Algorithm \ref{alg2}, for any $O \in \mathcal{I}$, if $\ell \ge \frac{1}{\eps}$, then:
    \begin{align}
        \sum_{j=1}^{\ell} f(T_j \mid \emptyset) \ge \left(1 - 2\eps\right) \cdot f(O) - \left(1 -\eps\right) \cdot f(\emptyset)
    \end{align}
\end{corollary}
\begin{proof} 
    \[
    \begin{array}{c}
    \sum_{j=1}^{\ell} (f(T_j) - f(\emptyset)) \ge \left(1 - \frac1\ell\right) f(O) - \frac1\ell \sum_{j=1}^{\ell} f(T_j) \\
    \Updownarrow \\
        \left(1 + \frac1\ell\right) \sum_{j=1}^{\ell} (f(T_j) - f(\emptyset)) \ge \left(1 - \frac1\ell\right) f(O) - f(\emptyset) \\
        \Updownarrow \\
        \begin{aligned}
        \sum_{j=1}^{\ell} (f(T_j) - f(\emptyset)) &\ge \frac{1 - \frac1\ell}{1 + \frac1\ell} f(O) - \frac{1}{1 + \frac1\ell} f(\emptyset) \\
        &\ge \left(1 - 2\eps\right) f(O) - \left(1 - \eps\right) f(\emptyset)
        \end{aligned}
    \end{array}
    \]
    The last inequality uses $\frac{1 - \eps}{1 + \eps} \ge 1 - 2\eps+\eps^2, \frac{1}{1+\eps}\ge 1-\eps-\eps^2$ for $\eps \in (0,1)$ and assumes $f(O) \ge f(\emptyset)$ for a simpler coefficient form. If $f(O) < f(\emptyset)$, the right-hand side is negative. However, dummy elements ensure $f(T_j) \ge f(\emptyset)$ since $f(T_j)$ is non-decreasing when elements are added. 
\end{proof}


\paragraph{Continuous Greedy Algorithm.} Inspired by \cite{DBLP:journals/corr/abs-2307-09616}, the authors illustrated the stationary point for non-monotone submodular maximization is probably arbitrary bad, which implies the solution far away from the bad stationary point has the potential to be a good solution. Leveraging this insight, we propose a deterministic continuous greedy algorithm that combines the information of the discrete stationary point through an auxiliary function. In Algorithm \ref{alg3}, we adopt the parameter $t_s$ to control the distance close to the stationary point. Before $t_s$, the feasible ascent directions is always orthogonal to the stationary point, i.e., always moving in the direction away from the stationary point. After $t_s$, we release the limitation for ascent direction, i.e., moving in all possible direction in feasible domain. Additionally, the concrete direction in each iteration is returned by the Algorithm \ref{alg2} based on the auxiliary function, while limiting the number of coordinates at the same time. The formal statement is as follows:


\begin{algorithm}[H]
  \caption{\textsc{ContinuousGreedy}($\mathcal{M}=(\mathcal{N},\mathcal{I})$, $f$, $\symz \subseteq \mathcal{N}$, $t_s \in [0,1]$, $\eps \in (0,1)$)}
  \label{alg3}
  \small
  $\delta \gets \eps^3,\quad \symy^0 \gets \mathbf{0}$\;
  \For{$i = 1$ \KwTo $1/\delta$}{
    $\symz_i \gets
      \begin{cases}
        \symz      & \text{if } i \in [1, t_s / \delta] \\
        \emptyset  & \text{if } i \in (t_s / \delta, 1 / \delta]
      \end{cases}$\;
    Define $f_{-\symz}(S) \triangleq f(S) - \sum_{u \in \symz \cap S}\bigl(f(\{u\}) - f(\emptyset) + 1\bigr)$\;
    Define $g_{-\symz}(S) \triangleq \mathbb{E}\bigl[f_{-Z}\bigl(\reali(\mathbf{e}_{S} \vee \symy^{i-1})\bigr)\bigr]$\;
    $(T_1,\dots,T_{1/\eps}) \gets \text{\textsc{Split}}(\mathcal{M}, g_{-\symz_i}, 1/\eps)$\;
    $\symx^i \gets \delta \cdot \sum_{j=1}^{1/\eps} \mathbf{e}_{T_j}$\;
    $\symy^{i} \gets \symy^{i-1} \oplus \symx^i$\;
  }
  \Return {$\symy^{1/\delta}$}\;
\end{algorithm}
We briefly explain the role of the modified set functions $f_{-\symz}$ and $g_{-\symz}$ used in Algorithm~\ref{alg3}. The purpose of these functions is purely technical: they ensure that elements in $\symz$ are never selected by the split procedure, while preserving the relative marginal values of all other elements. Specifically, for any $u \in \symz$, the definition of $f_{-\symz}$ enforces a strictly negative marginal contribution for adding $u$ to any set, and hence $u$ is dominated by every element outside $\symz$ in the greedy selection. For elements $v \notin \symz$, the marginal gains $f_{-\symz}(v \mid S)$ coincide with $f(v \mid S)$ for all $S \subseteq \mathcal{N}$, so the behavior of the objective function on $\mathcal{N} \setminus \symz$ remains unchanged. The function $g_{-\symz}$ is the corresponding extended multilinear extension evaluated
at the current fractional point. Importantly, the support of $f_{-\symz}$ and $g_{-\symz}$ differs from that of
$f$ only by the elements in $\symz$, and evaluating $g_{-\symz}$ incurs the same asymptotic computational cost as evaluating the original function, up to a constant-factor overhead.

\begin{lemma}
    The set functions $f_{-\symz}(S)$ and $g_{-\symz}(S)$ defined in Algorithm \ref{alg3} are submodular. Note we do not require these functions to be non-negative.
\end{lemma}
\begin{proof}
    For any $A \subseteq B \subseteq \mathcal{N}, u \in \mathcal{N} \setminus B$:
    \begin{align*}
        f_{-\symz}(u \mid B) &= f(u \mid B) - [u \in \symz] (f(\{u\}) - f(\emptyset) + 1) \\
        &\le f(u \mid A) - [u \in \symz] (f(\{u\}) - f(\emptyset) + 1) = f_{-\symz}(u \mid A)
    \end{align*}
    Thus $f_{-\symz}$ is submodular. Next,
    \begin{align*}
        g_{-\symz}(u \mid B) &= \mathbb{E} [f_{-\symz}(\reali(\mathbf{e}_{B + u} \vee \symy^{i-1}))] - \mathbb{E} [f_{-\symz}(\reali(\mathbf{e}_{B} \vee \symy^{i-1}))] \\
        &= \mathbb{E} \left[ f_{-\symz}(\reali(\symy^{i-1}) \cup (B + u)) - f_{-\symz}(\reali(\symy^{i-1}) \cup B) \right] \\
        &\le \mathbb{E} \left[ f_{-\symz}(\reali(\symy^{i-1}) \cup (A + u)) - f_{-\symz}(\reali(\symy^{i-1}) \cup A) \right] \quad \text{(submodularity)} \\
        &= g_{-\symz}(u \mid A)
    \end{align*}
    The second equality treats both $\reali(\symy^{i-1})$ occurrences as the same random variable. The inequality uses the submodularity of $f_{-\symz}$. Thus $g_{-\symz}$ is submodular.
\end{proof}
\begin{lemma}
    For any $u \in Z$, the function $g_{-Z}(S)$ defined in Algorithm \ref{alg3} satisfies for any set $S \subseteq \mathcal{N} - u$: $ g_{-Z}(S) > g_{-Z}(S + u)$.
\end{lemma}
\begin{proof}
    First, we claim: For any $T \subseteq \mathcal{N} - u$, $f_{-\symz}(T) > f_{-\symz}(T + u)$. This follows from the definition of $f_{-\symz}$ and submodularity:
    \begin{align*}
        f_{-\symz}(T + u) - f_{-\symz}(T) &= (f(T + u) - f(T)) + (-1 \cdot (f(\{u\}) - f(\emptyset) + 1) \\
        &\le (f(\{u\}) - f(\emptyset)) - (f(\{u\}) - f(\emptyset) + 1) = -1
    \end{align*}
    Now, note that $\reali(\mathbf{e}_{S} \vee \symy) = \reali(\symy) \cup S$ because elements in $S$ are always included in $\reali(\mathbf{e}_{S} \vee \symy)$, and the values at indices corresponding to $S$ do not affect the distribution of other elements in $\reali(\mathbf{y})$.
    Therefore,
    \begin{align*}
        g_{-Z}(S) &= \mathbb{E} [f_{-\symz}(\reali(\mathbf{y}) \cup S)] = \sum_{T \subseteq \mathcal{N}} \Pr[\reali(\mathbf{y}) = T] \cdot f_{-\symz}(T \cup S) \\
        &> \sum_{T \subseteq \mathcal{N}} \Pr[\reali(\mathbf{y}) = T] \cdot f_{-\symz}(T \cup (S + u)) \quad \text{(by claim above)} \\
        &= \mathbb{E} [f_{-\symz}(\reali(\mathbf{y}) \cup (S + u))] = g_{-Z}(S + u)
    \end{align*}
\end{proof}
Due to the negative incremental gain property of elements $u \in Z$, and the fact that dummy elements have zero incremental gain and are sufficiently numerous, such $u$ will never be selected by the \textsc{Split}(Algorithm \ref{alg2}) into any set.
\begin{corollary}
    \label{notinclude}
    In the $i$-th iteration of Algorithm \ref{alg3}, for any $j \in [1/\eps]$, $T_{j} \cap \symz_i = \emptyset$.
\end{corollary}
Next, we establish several structural properties of the sequence $\{\symy^i\}_{i=0}^{1/\delta}$ generated by Algorithm~\ref{alg3}. First, the solution constructed by the algorithm is feasible, in particular that the final marginal vector lies in the matroid polytope. 
\begin{lemma}
    \label{marinP}
    $\Mar(\symy^{1/\delta}) \in \mathcal{P}(\mathcal{M})$.
\end{lemma}
\begin{proof}
    Since the sets returned by \textsc{Split} form a basis of $\mathcal{M}$, $\frac{1}{\delta} \Mar(\symx^i) \in \mathcal{P}(\mathcal{M})$. Therefore, their linear combination $\sum_{i \in [1/\delta]} \symx^i \in \mathcal{P}(\mathcal{M})$.
    Thus,
    \begin{align*}
        \Mar(\symy^{1/\delta}) = \Mar\left( \bigoplus_{i \in [1/\delta]} \symx^i \right) = \bigoplus_{i \in [1/\delta]} \Mar\left( \symx^i \right)
    \end{align*}
    For any $u \in \mathcal{N}$, $\Mar(\symy^{1/\delta})_u = \bigoplus_{i \in [1/\delta]} \Mar(\symx^i)_u \le \left( \sum_{i \in [1/\delta]} \symx^i \right)_u$.
    By the downward-closed property of $\mathcal{P}(\mathcal{M})$, it follows that $\Mar(\symy^{1/\delta}) \in \mathcal{P}(\mathcal{M})$.
\end{proof}

Additionally, we show that the support of $\symy^i$ grows in a controlled manner, that all marginal values remain bounded, and that the contribution of elements in $\symz$ is appropriately delayed. These bounds allow us to interpret the update rule as a valid continuous greedy process and will be repeatedly invoked in the proof of the main approximation result.
\begin{lemma}
    \label{suppproof}
    For any integer $0 \le i \le 1/\delta$, $\supp (\symy^{ i}) \le (1/\eps) \cdot i \le 1/\eps^4$, $\Vert \Mar(\symy^{ i}) \Vert_{\infty} \le 1 - (1 - \delta)^i$, and $\Vert \Mar(\symy^{ i}) \wedge \mathbf{1}_{\symz} \Vert_{\infty} \le 1 - (1 - \delta)^{\max\{0, i - t_s / \delta\}}$.
\end{lemma}
\begin{proof}
    Since each update from $\symy^{i}$ to $\symy^{i+1}$ affects only $1/\eps$ positions, the number of non-zero positions in $\symy^{ i}$ is at most $(1/\eps) \cdot i$. Additionally, the sets $(T_1, T_2, \dots, T_{1/\eps})$ returned by \textsc{Split} are disjoint, i.e., $\Vert \Mar(\symx^i) \Vert_{\infty} \le \delta$.
    Using the property of $\Mar$ and $\oplus$ (Observation \ref{maroplus}), for any $u \in \mathcal{N}$:
    \begin{align*}
        \Mar(\symy^{ i})_u = \Mar\left(\symy^{i-1}\right)_u \oplus \Mar\left(\symx_i\right)_u \le 1 - \left(1 - \Mar(\symy^{i-1})_u\right)(1 - \delta)
    \end{align*}
    Recursively applying this gives $\Mar(\symy^{ i})_u \le 1 - (1 - \delta)^i$.

    By Corollary \ref{notinclude}, for any $u \in \symz, i \le t_s / \delta$, we have $u \notin T_1 \cup T_2 \cup \dots T_{1/\eps}$, so $\Mar(\symy^{ i})_u = 0$. Then, applying the above inequality yields for $u \in \symz, i \ge t_s / \delta$:
    $\Mar(\symy^{i})_u \le 1 - (1 - \delta)^{i - t_s / \delta}$.
\end{proof}
\begin{lemma}
    \label{upperboundofmar}
    For any $i \le 1/\delta$, we have $1 - (1 - \delta)^i \le 1 - e^{-i\delta} + i\delta^2$, $\Vert\Mar(\symy^i)\Vert_{\infty} \le 1 - e^{-i\delta} + \delta$, and $\Vert\Mar(\symy^i) \wedge \mathbf{1}_{\symz}\Vert_{\infty} \le 1 - e^{-\max\{0, i\delta - t_s\}} + \delta$.
\end{lemma}
\begin{proof}
    Prove by induction on $i$. For $i=0$, it holds trivially. For $i=1$, $1 - (1 - \delta) \le 1 - e^{-\delta} + \delta^2$. Expanding $e^{-\delta}$:
    \begin{align*}
        e^{-\delta} - (1 - \delta + \delta^2) &= \sum_{k \ge 0} \frac{(-\delta)^k}{k!} - (1 - \delta + \delta^2) \\
        &= -\frac{\delta^2}{2} - \sum_{k \ge 2} \left( \frac{\delta^{2k-1}}{(2k-1)!} - \frac{\delta^{2k}}{(2k)!} \right) \le 0
    \end{align*}
    So it holds for $i=1$. Assume it holds for all integers in $[0, i]$, prove for $i+1$:
    \begin{align*}
        1 - (1 - \delta)^{i+1} &= (1 - (1 - \delta)^i)(1 - \delta) + \delta \\
        &\le (1 - e^{-i\delta} + i\delta^2)(1 - \delta) + \delta \\
        &= 1 + i\delta^2 - i\delta^3 - e^{-i\delta}(1 - \delta) \\
        &\le 1 + i\delta^2 - i\delta^3 - e^{-i\delta}(e^{-\delta} - \delta^2) \quad \text{(using $i=1$ case)} \\
        &= 1 + (i\delta^2 + e^{-i\delta} \delta^2) - e^{-(i+1)\delta} - i\delta^3 \\
        &\le 1 + (i+1)\delta^2 - e^{-(i+1)\delta} \quad \text{(since $e^{-i\delta} \le 1$)}
    \end{align*}
\end{proof}
\begin{lemma}
    \label{polytime}
    Algorithm \ref{alg3} requires a total of $O_\eps(n r)$ queries.
\end{lemma}
\begin{proof}
    Since $\reali(\symy)$ has at most $2^{\ffrac(\symy)}$ distinct possible values, and by Lemma \ref{suppproof}, $\ffrac(\symy^i) \le \supp(\symy^i) \le 1/\eps^4$ (a constant depending only on $\eps$), computing $g_{-\symz}(S)$ requires a constant number of value queries.
    In Algorithm \ref{alg3}, only the \textsc{Split} subroutine requires query complexity beyond $\eps$ dependence. One \textsc{Split} execution requires $O(n r)$ computations of $g_{-\symz}(S)$ and independence queries. Therefore, Algorithm \ref{alg3} requires $O_\eps(n r)$ queries in total.
\end{proof}

Leveraging the above lemmas for the sequence $\{\symy^i\}_{i=0}^{1/\delta}$, we can obtain a lower bound on the function value growth at each step.
\begin{lemma}
    \label{lemmaofSPLIT}
    Let $OPT_i = OPT \setminus \symz_i$. In Algorithm \ref{alg3}, for any $i \in [0, 1/\delta)$,
    \begin{align*}
        \frac{1}{\delta} \left[ F(\symy^{ i}) - F(\symy^{i-1}) \right] \ge (1 - 3\eps) \cdot \left( F(\mathbf{e}_{OPT_i} \vee \symy^{i-1}) - F(\symy^{i-1}) \right)
    \end{align*}
\end{lemma}
\begin{proof}
    Note that for $i \le t_s / \delta$, the algorithm does not involve any elements in $\symz$. It can be viewed as executing on the set $\mathcal{N} \setminus \symz$. Since $T_j \cap \symz_i = \emptyset$, $OPT_i \cap \symz_i = \emptyset$, $\reali(\symy^{i}) \cap \symz_i = \emptyset$, we easily get $F(\mathbf{e}_{OPT_i} \vee \symy^{i}) = g_{-\symz_i}(OPT_i)$, $F(\mathbf{e}_{T_j} \vee \symy^{i}) = g_{-\symz_i}(T_j)$.
    
    Using the property of Algorithm \textsc{Split} (Corollary \ref{splitprop}), we have:
    \begin{align}
        \label{aligninSPLIT}
        \sum_{j=1}^{1/\eps} \left( F(\mathbf{e}_{T_j} \vee \symy^{i-1}) - F(\symy^{i-1}) \right) \ge (1 - 2\eps) F(\mathbf{e}_{OPT_i} \vee \symy^{i-1}) - (1 - \eps) F(\symy^{i-1})
    \end{align}
    Now note:
    \begin{align*}
        F(\symy^i) &= \mathbb{E} [f(\reali(\symy^i))] = \mathbb{E} [f(\reali(\symy^{i-1} \oplus \symx^i))] = \mathbb{E} [f(\reali(\symy^{i-1}) \cup \reali( \symx^i))] \\
        &= \sum_{S \subseteq \mathcal{N}} \Pr[\reali(\symx^i) = S] \cdot \mathbb{E} [f(\reali(y^{i-1}) \cup S)] \\
        &= \sum_{S \subseteq \mathcal{N}} \Pr[\reali(\symx^i) = S] \cdot F(\mathbf{e}_{S} \vee \symy^{i-1}) \\
        &\ge (1 - \delta)^{1/\eps} F(\symy^{i-1}) + \sum_{j=1}^{1/\eps} \delta (1 - \delta)^{1/\eps - 1} F(\mathbf{e}_{T_j} \vee \symy^{i-1}) \quad \text{(only $S=\emptyset$ or $T_j$)} \\
        &\ge \left(1 - \frac{\delta}{\eps}\right) F(\symy^{i-1}) + \delta \left(1 - \frac{\delta}{\eps}\right) \sum_{j=1}^{1/\eps} F(\mathbf{e}_{T_j} \vee \symy^{i-1}) \quad \text{($(1-\delta)^{k} \ge 1 - k\delta$)} \\
        &\ge \left(1 - \frac{\delta}{\eps}\right) \left[ F(\symy^{i-1}) + \delta \left( (1 - 2\eps) F(\mathbf{e}_{OPT_i} \vee \symy^{i-1}) + \left(\frac{1}{\eps} - 1 + \eps\right) F(\symy^{i-1}) \right) \right] \quad \text{(using (\ref{aligninSPLIT}))} \\
        &\ge (1 - \delta) F(\symy^{i-1}) + \delta (1 - 3\eps) F(\mathbf{e}_{OPT_i} \vee \symy^{i-1}) \quad \text{(using $\delta = \eps^3$ and simplifying)}
    \end{align*}
    Rearranging gives the result.
\end{proof}
Note that the lower bound in the above lemma involves $F(\mathbf{e}_{OPT_i} \vee \symy^{i-1})$. In the next lemma, we analyze this further to relate it to the optimal solution $f(OPT)$.
\begin{lemma}
    \label{lemmaofincreament}
    In Algorithm \ref{alg3}, for any $i \in [1, 1/\delta]$, $A \subseteq \mathcal{N}$, we have:
    \begin{align*}
        F(\symy^i \vee \mathbf{e}_A) \ge \left( e^{-\max\{0, t - t_s\}} - e^{-t} \right) \cdot  \max \{0, f(A) - f(A \cup Z) \} + (e^{-t} - \delta) \cdot f(A)
    \end{align*}
    where $t = i \delta$.
\end{lemma}
\begin{proof}
    By Lemma \ref{randomset}, $F(\symy^i \vee \mathbf{e}_{A}) \ge \hat{f}(\Mar(\symy^i \vee \mathbf{e}_{A}))$.
    Let $\theta^{\symz} = e^{-\max\{0, t - t_s\}} - \delta$, $\theta = e^{-t} - \delta$.
    By Lemma \ref{upperboundofmar}, for $u \in \symz$, $\Mar(\symy^i)_u \le 1 - \theta^{\symz}$; for $u \notin \symz$, $\Mar(\symy^i)_u \le 1 - \theta$.
    Note that $\Mar(\symy^i \vee \mathbf{e}_A) = \Mar(\symy^i) \vee \mathbf{1}_A$ because $\reali(\symy^i \vee \mathbf{e}_A) = \reali(\symy^i) \cup A$. Using these observations:
    \begin{align*}
        F(\symy^i \vee \mathbf{e}_{OPT}) &\ge \hat{f}(\Mar(\symy^i) \vee \mathbf{1}_A) = \int_0^1 f(T_{\lambda}(\Mar(\symy^i)) \cup A)  d\lambda \\
        &\ge \int_{1 - \theta^{\symz}}^{1 - \theta} f(T_{\lambda}(\Mar(\symy^i)) \cup A)  d\lambda + \int^{1}_{1 - \theta} f(T_\lambda(\Mar(\symy^i)) \cup A)  d\lambda \\
        &\ge \int_{1 - \theta^{\symz}}^{1 - \theta} f(T_{\lambda}(\Mar(\symy^i)) \cup A)  d\lambda + \theta \cdot f(A) \quad \text{(since $T_\lambda = \emptyset$ for $\lambda > 1 - \theta$)} \\
        &\ge (\theta^{\symz} - \theta) \cdot \max\{f(A) - f(Z \cup A), 0\} + \theta \cdot f(A)
    \end{align*}
    The last inequality holds because when $1 - \theta^{\symz} < \lambda \le 1 - \theta$, we have $T_\lambda(\Mar(\symy^i)) \subseteq \mathcal{N} \setminus \symz$. By submodularity and non-negativity of $f$, for any $B \subseteq \mathcal{N} \setminus \symz$:
    \begin{align*}
        &f(B \cup (\symz \setminus A) \cup A) - f(B \cup A) \le f((Z \setminus A) \cup A) - f(A) \\
        \Rightarrow& f(B \cup A) \ge f(A) - f(Z \cup A) \\
        \Rightarrow& \int_{1 - \theta^{\symz}}^{1 - \theta} f(T_{\lambda}(\Mar(\symy^i)) \cup A)  d\lambda \ge (\theta^{\symz} - \theta) \cdot (f(A) - f(Z \cup A))
    \end{align*}
    And this value is also $\ge 0$ due to non-negativity.
\end{proof}

To translate the per-iteration improvement bounds into a bound on the final
value $F(\symy^{1/\delta})$, we compare the discrete evolution of
$F(\symy^{i})$ with an auxiliary function. Since $A$ is arbitrary in the above lemma, we can finally express our solution $F(\symy^{1/\delta})$ in terms of $OPT$ and $Z$, and then complete the entire algorithm via Theorem \ref{rounding}. Define $g(0) = 0$,
\begin{align*}
    g((i+1)\delta) = \begin{cases}
        g(i\delta) + \delta \left[ f(OPT \setminus \symz) - (1 - e^{-i\delta}) \cdot f(\symz \cup OPT) - g(i\delta) \right] & \text{if } i\delta < t_s \\
        g(i\delta) + \delta \left[ e^{-i\delta} \cdot f(OPT) + (e^{t_s - i\delta} - e^{-i\delta}) \cdot \max\{f(OPT) - f(\symz \cup OPT), 0\} - g(i\delta) \right] & \text{if } i\delta \ge t_s
    \end{cases}
\end{align*}
\begin{lemma}
    \label{fandg}
    For any $i \in [0, 1/\delta]$, $g(i\delta) \le F(\symy^i) + 4i\delta \eps \cdot f(OPT)$
\end{lemma}
\begin{proof}
    We prove by induction on $i$. For $i=0$, $g(0) = 0 \le F(\symy^0)$. Assume the lemma holds for $i$, prove for $i+1$.
    Define the increment of $g$ as:
    \begin{align*}        
        g'(i\delta) = \begin{cases}
            f(OPT \setminus \symz) - (1 - e^{-i\delta}) \cdot f(\symz \cup OPT) & \text{if } i\delta < t_s \\
            e^{-i\delta} \cdot f(OPT) + (e^{t_s - i\delta} - e^{-i\delta}) \cdot \max\{f(OPT) - f(\symz \cup OPT), 0\} & \text{if } i\delta \ge t_s
        \end{cases}
    \end{align*}
    By definition, $g'(i\delta) \le f(OPT)$ and
    \begin{align*}
        g((i+1)\delta) = \delta \left[ g'(i\delta) - g(i\delta) \right] + g(i\delta)
    \end{align*}
    By Lemmas \ref{lemmaofSPLIT} and \ref{lemmaofincreament}, we have:
    \begin{align}
        \label{fincreament}
        F(\symy^{i+1}) \ge F(\symy^i) + \delta \left[ (1 - 3\eps) (g'(i\delta) - \delta \cdot f(OPT)) - F(\symy^i) \right] \\
        \ge F(\symy^i) + \delta \left[ (1 - 3\eps) g'(i\delta) - \delta \cdot f(OPT) - F(\symy^i) \right]
    \end{align}
    Therefore:
    \begin{align*}
        g((i+1)\delta) &= g(i\delta) + \delta (g'(i\delta) - g(i\delta)) \\
        &= (1 - \delta) g(i\delta) + \delta g'(i\delta) \\
        &\le (1 - \delta) [ F(\symy^{i}) + 4i\delta\eps \cdot f(OPT) ] + \delta g'(i\delta) \quad \text{(induction hyp.)} \\
        &\le F(\symy^{i+1}) - \delta (1 - 3\eps) g'(i\delta) + \delta^2 f(OPT) + (1 - \delta) \cdot 4i\delta\eps \cdot f(OPT) + \delta g'(i\delta) \quad \text{(using \ref{fincreament})} \\
        &= F(\symy^{i+1}) + (1 - \delta)4i\delta\eps \cdot f(OPT) + 3\eps\delta g'(i\delta) + \delta^2 f(OPT) \\
        &\le F(\symy^{i+1}) + 4((i+1)\delta)\eps \cdot f(OPT) \quad \text{(since $g'(i\delta) \le f(OPT), \delta < \eps$)}
    \end{align*}
\end{proof}
\begin{theorem}
    \label{thm2.4}
    The solution returned by Algorithm \ref{alg3} satisfies the following inequality:
    \begin{align*}        
        F(\symy^{1/\delta}) \ge e^{t_s - 1} \Big[ (2 - t_s - e^{-t_s} - 4\eps ) f(OPT) &- (1 - e^{-t_s}) f(\symz \cap OPT) \\
        &- (2 - t_s - 2e^{-t_s}) f(\symz \cup OPT) \Big]
    \end{align*}
\end{theorem}
\begin{proof}
    The definition of $g(i\delta)$ is consistent with that in \cite{DBLP:conf/stoc/BuchbinderF24}. Its Corollary A.5 proves:
    \begin{align*}
        g(1) \ge e^{t_s - 1} \Big[ (2 - t_s - e^{-t_s}) f(OPT) &- (1 - e^{-t_s}) f(\symz \cap OPT) \\
        &- (2 - t_s - 2e^{-t_s}) f(\symz \cup OPT) \Big]
    \end{align*}
    By Lemma \ref{fandg},
    \begin{align*}
        F(\symy^{1/\delta}) &\ge g(1) - 4\eps f(OPT) \\
        &\ge e^{t_s - 1} \Big[ (2 - t_s - e^{-t_s} - 4\eps) f(OPT) - (1 - e^{-t_s}) f(\symz \cap OPT) \\
        &\qquad \qquad - (2 - t_s - 2e^{-t_s}) f(\symz \cup OPT) \Big]
    \end{align*}
\end{proof}

\subsection{Proof of the Main Result}

In the previous subsections, we established the theoretical guarantees for the two main components of our algorithm. 
First, the \textsc{Split} algorithm returns a set $\symz$ satisfying the structural inequalities in Theorem~\ref{thm2.1}. 
Second, the continuous greedy phase produces a fractional vector whose value satisfies the bound in Theorem~\ref{thm2.4}. 
Finally, by Theorem~\ref{rounding}, the \textsc{Deterministic-Pipage} procedure converts any fractional vector $y$ with $\mar(y)\in\mathcal{P}(\mathcal{M})$ into a feasible set $S\in\mathcal{M}$ such that $f(S)\ge F(y)$.

Therefore, to prove the approximation guarantee of our algorithm, it suffices to establish a lower bound on the value returned by Algorithm~\ref{alg4}. 
We now prove the main theorem.

\paragraph{Proof of Theorem \ref{mainthm}.}

By Theorem \ref{rounding}, the \textsc{Deterministic-Pipage} algorithm converts a fractional vector $y$ with $\mar(y)\in \mathcal{P}(\mathcal{M})$ into a feasible set $S\in\mathcal{M}$ satisfying
\[
f(S) \ge F(y).
\]
Therefore, it suffices to prove a lower bound on the fractional value returned by Algorithm~\ref{alg4}.

The feasibility of the solutions produced by the algorithm and the query complexity bound follow from Theorem \ref{thm2.1}, Lemma \ref{marinP}, and Lemma \ref{polytime}. Hence, we only need to analyze the approximation ratio.

The lower bound for the solution of our algorithm follows essentially the same analysis as in \cite{DBLP:journals/mor/BuchbinderF19}. The algorithm returns the maximum of two candidate solutions. We first list the inequalities satisfied by these two quantities.

By Theorem \ref{thm2.1}, for every $A \in \mathcal{I}$,
\begin{align*}
    f(\symz) \ge \frac{1}{2} f(\symz \cup A) + \frac{1}{2} f(\symz \cap A) - \eps \cdot f(OPT).
\end{align*}

Substituting $A = OPT$ and $A = OPT \cap \symz$ gives
\begin{align}
    \label{eq5}
    f(\symz) \ge \frac{1}{2} f(\symz \cup OPT) + \frac{1}{2} f(\symz \cap OPT) - \eps \cdot f(OPT),
\end{align}
\begin{align}
    \label{eq6}
    f(\symz) \ge f(\symz \cap OPT) - \eps \cdot f(OPT).
\end{align}

By Theorem \ref{thm2.4},
\begin{align}
    \label{eq7}
    F(\symy^{1/\delta}) \ge e^{t_s - 1} \Big[ (2 - t_s - e^{-t_s} - 4\eps ) f(OPT)
    - (1 - e^{-t_s}) f(\symz \cap OPT) 
    - (2 - t_s - 2e^{-t_s}) f(\symz \cup OPT) \Big].
\end{align}

Inequalities \eqref{eq5}, \eqref{eq6}, and \eqref{eq7} are all valid lower bounds on the objective value returned by the algorithm. Therefore, any convex combination of them is also a valid lower bound.

Let $ALG$ denote the fractional solution returned by the algorithm. For any $p_1,p_2,p_3 \ge 0$ with $p_1+p_2+p_3=1$, we obtain
\begin{align*}
    F(ALG) \ge& 
    \frac{1}{2} p_1 \left( f(\symz \cup OPT) + f(\symz \cap OPT) + \eps f(OPT) \right) \\
    &+ p_2 \left( f(\symz \cap OPT) + \eps f(OPT) \right) \\
    &+ p_3 e^{t_s - 1} \Big[ (2 - t_s - e^{-t_s} - 4\eps ) f(OPT) \\
    &\quad - (1 - e^{-t_s}) f(\symz \cap OPT)
    - (2 - t_s - 2e^{-t_s}) f(\symz \cup OPT) \Big].
\end{align*}

We choose $p_1,p_2,p_3$ and $t_s$ so that the coefficients of $f(\symz \cup OPT)$ and $f(\symz \cap OPT)$ are non-negative while maximizing the coefficient of $f(OPT)$.

As shown in Section~3.1 of \cite{DBLP:journals/mor/BuchbinderF19}, the optimal parameters are
\[
p_1 = 0.205,\quad p_2 = 0.025,\quad p_3 = 0.070,\quad t_s = 0.372,
\]
which yields
\[
F(ALG) \ge (0.385 - O(\eps))\, f(OPT).
\]

Applying Theorem \ref{rounding} completes the proof.

\section{Knapsack Constraints}
In this section, we present our deterministic algorithm for the knapsack constraint and analyze its approximation ratio and query complexity. Our algorithm consists of three components: enumerate elements, optimization, and rounding. The algorithm first enumerates all subsets $E_i$ of the ground set whose size is at most $\eps^{-2}$. Then, for each $E_i$, we assume $E_i$ has been selected, that is, we consider $f(\cdot \mid E_i)$ as the new function and $(1-\eps)B$ as the capacity. We apply a measured continuous greedy algorithm on the extended multilinear extension to obtain a vector $\mathbf{y}_i \in [0,1]^{2^\N}$. Next, we apply a rounding algorithm to $\mathbf{y}_i$ to obtain a set $S_i$. Finally, we select the set with the largest value among all $S_i$ as our output. The above procedure is summarized in Algorithm~\ref{Knapsack:whole}.


\begin{algorithm}[H]
    \caption{Deterministic Knapsack$(f,w,B)$}
        \label{Knapsack:whole}
    \For{
        each set $E_i\in \mathcal{N}$ with $|E_i| \le \eps^{-2}$
    }{
        $\mathbf{y}_i \gets DMCG(f(\cdot\cup E_i), w, (1 - \eps)(B - w(E_i)), \eps)$ \\ 
        $S_i \gets \textsc{Rounding}(f,w, \mathbf{y}_i)$
    }
    return $arg \max_i f(S_i \cup E_i)$
\end{algorithm}

\subsection{Enumerate Elements}
We begin our analysis with the element enumeration step. Following the approach in \cite{DBLP:journals/siamcomp/ChekuriVZ14}, by enumerating all subsets of size at most $O(\varepsilon^{-2})$, we effectively fix the heavy elements and reduce the remaining problem to one consisting only of lightweight elements, each contributing at most an $\varepsilon$ fraction of the total weight. We then show that, in this reduced instance, the knapsack capacity can be safely scaled down by a factor of $\alpha$, and under this reduced capacity, the value of the optimal solution decreases by at most $\alpha - O(\varepsilon)$. This reduction has two purposes: first, it makes the knapsack constraint behave more like a cardinality constraint during the optimization phase; second, it provides slack for the subsequent rounding procedure to ensure feasibility.
\begin{lemma}[Corollary 4.16 in \cite{DBLP:journals/siamcomp/ChekuriVZ14}]
\label{lemma:enumerate_basic} For maximizing a submodular function $f: 2^{\mathcal N} \to \mathbb{R}^+ $ under a knapsack constraint with capacity $B$ and a weight function
$ w: \mathcal N \to \mathbb{R}^+ $,
there exists a subset $E \subseteq \N $ satisfying the following properties:

1. Its cardinality is bounded by $|E| \le \varepsilon^{-2}$, and hence $E$ can be identified by enumerating at most $\varepsilon^{-2}$ elements.

2. For every element $o \in \OPT \setminus E$, the marginal contribution satisfies

   \[ 
   f(o \mid E) \le \varepsilon^{2} f(\OPT).
    \]
   
3. Let $B_{ig}  = \{u \in \OPT|w(u) \ge \eps (B - w(E))\}\backslash E$, then $f(\OPT \backslash B_{ig}) \ge (1 - \eps)f(OPT)$.
    
\end{lemma}
\begin{proof}
We consider a subset $E$ constructed as follows. Initially, $E$ is empty. We then iterate over the elements of the optimal solution $\mathrm{OPT}$. For each element $o \in \mathrm{OPT}$, if

\[ 
f(o \mid E) \ge \varepsilon^{2} f(\mathrm{OPT}),
 \]
 then we pick the item $o$ into $E$.
We first observe that if $|E| > \varepsilon^{-2}$, then by the construction rule of $E$ we would have
\[ 
f(E) > |E| \cdot \varepsilon^{2} f(\OPT) \ge \varepsilon^{-2} \cdot \varepsilon^{2} f(\OPT) = f(\OPT).
 \]
This is impossible, since $f(E) \le f(OPT)$ denotes the optimal value. Therefore, we must have $|E| \le \varepsilon^{-2}$.

Next, consider the second property. Since $f$ is submodular, the marginal contribution of any element can only decrease as the set $E$ grows. Consequently, once every element of $\OPT$ has been examined exactly once, submodularity guarantees that its marginal contribution in all subsequent steps cannot increase, and hence no remaining element can satisfy the inequality above. Therefore, for all $o \in \OPT \setminus E$, we have
\[
f(o \mid E) \le \varepsilon^{2} f(\OPT).
\]

Finally, consider the third property. Let $B_{ig} = \{b_1, b_2, \ldots, b_t\}$ denote the subset of elements in $\OPT$ with weights at least $\varepsilon (B - w(E))$. We analyze the process of removing these elements from $\OPT$ one by one. For each $b_i \in B_{ig}$, the loss incurred is bounded by
\[ 
f(b_i \mid \OPT \setminus B_{ig} \cup \{b_{i+1}, \ldots, b_t\}) \le f(b_i \mid \OPT \setminus B_{ig}) \le f(b_i \mid E) \le \varepsilon^{2} f(\OPT),
 \]
where the last inequality follows from the property that no element outside $E$ has marginal contribution larger than $\varepsilon^{2} f(\OPT)$. Since there are at most $\varepsilon^{-1}$ such elements, the total loss is bounded by $\varepsilon f(\OPT)$.
\end{proof}

Moreover, we prove that if we assume such an element $E$ has already been selected and further reduce the remaining capacity to a $\alpha$ fraction of the original, then the value of the optimal solution decreases by at most a factor of $(\alpha-\varepsilon^{2})$.

\begin{lemma}
\label{lemma:enumurate_2}
Given the problem of maximzing submodular function $f: 2^{\mathcal N} \to \mathbb{R}^+ $ under a knapsack constraint with capacity $B$ and a weight function 
$ w: \mathcal N \to \mathbb{R}^+ $.  If for any $u \in \N$, we have $f(u \mid \emptyset) \le \eps^2f(\OPT)$ and $w(u) \le \eps B$, then we have 
\[ 
\max_{w(S)\le \alpha B} f(S) \ge (\alpha - \varepsilon^{2}) f(\mathrm{OPT}).
 \]
\end{lemma}
\begin{proof}
We consider the following iterative procedure. 
Let $O_0 = \mathrm{OPT}$. 
At the $i$-th round, we select the element $o_i \in O_{i-1}$ with the lowest current density, i.e.,
\[
o_i = \arg\min_{u \in O_{i-1}} \frac{f(u \mid O_{i-1} - u)}{w(u)},
\] 
and remove it to obtain $O_i = O_{i-1} \setminus \{o_i\}$, continuing until $w(O_t) \le \alpha B$ for some $t$. Since $o_i$ has the lowest density, it follows that
\[
f(o_i \mid  O_{i-1} - o_i)  \le f(O_{i-1}) \cdot \frac{w(o_i)}{w(O_{i-1})}.
\]
Otherwise, summing over all $u \in O_{i-1}$, we would obtain
\[
f(O_{i-1}) \ge \sum_{u \in O_{i-1}}f(u\mid O_{i-1} - u) > f(O_{i-1})\sum_{u \in O_{i-1}}\frac{w(u)}{w(O_{i-1})} = f(O_{i-1}),
\]
which is a contradiction. Rearranging terms, we then have
\[
\frac{f(O_{i})}{w(O_{i})} \ge \frac{f(O_{i - 1})}{w(O_{i - 1})}.
\]
By applying this inequality iteratively, we obtain
\[
\frac{f(O_{t - 1})}{w(O_{t-1})} \ge \frac{f(O_{t - 2})}{w(O_{t-2})} \ge \cdots \ge \frac{f(\mathrm{OPT})}{w(\mathrm{OPT})}.
\]
It follows that
\begin{align*}
f(O_t) &\ge f(O_{t - 1}) - f(o_t \mid O_{t}) \ge f(\mathrm{OPT}) \frac{w(O_{t - 1} )}{w(\mathrm{OPT})} - f(o_t \mid O_{t}) \\
    &\ge \alpha f(\mathrm{OPT}) - f(o_t \mid O_{t}) \ge \alpha f(\mathrm{OPT}) - f(o_t \mid \emptyset) \ge (\alpha - \varepsilon^2)\cdot f(\mathrm{OPT}).
\end{align*}
\end{proof}




Consequently, our algorithm first enumerates all subsets $E_i$ of size at most $\varepsilon^{-2}$. 
This guarantees that the desired set $E$ will be included among the enumerated candidates. 
For each such set $E_i$, we solve the corresponding derived optimization problem defined above using an appropriate optimization algorithm. 
If this algorithm returns an approximate fractional solution for the derived problem, then, when applied to the correct set $E$, it yields a fractional solution that achieves the desired approximation guarantee.

\subsection{Deterministic Measured Continuous Greedy}
We now describe our optimization component, which relies on a novel split algorithm namely, \textsc{KnapsackSplit} subroutine. Different from the matroid split algorithm, we replace the marginal gain $f(u\mid T_j)$ with the marginal density $f(u\mid T_j)/w(u)$. 
Throughout the remainder of this section, we set $\ell = 1/\varepsilon$ and, without loss of generality, assume that $1/\varepsilon$ is an integer. 
For the \textsc{KnapsackSplit} algorithm, we establish the following result:

\begin{lemma} 
\label{lemma:knapsack_split}
Given a submodular function $f: 2^{\mathcal N} \to \mathbb{R}^+ $, a knapsack constraint with capacity $B$ and a weight function 
$ w: \mathcal N \to \mathbb{R}^+ $, then for any feasible set $Q$, if $f(u \mid \emptyset) \le \eps^2f(Q) $ for all $ u \in \NN$,
 the \textsc{KnapsackSplit} Algorithm produces disjoint sets $T_1, T_2, \ldots, T_{\ell} \subseteq \mathcal{N}$ such that their union $T=\cup_{j=1}^{\ell} T_j$ that satisfying $\cup_{j = 1}^{\ell}T_j \le B$, and then we have

\[ 
\sum_{j=1}^{\ell} f\left(T_j \mid \varnothing\right) \geq \max\left\{\frac{1}{\ell}\sum_{j = 1}^{\ell}\left( f(Q\cup T_j) - f(T_j)\right) - \eps^2f(OPT), 0 \right\}.
 \]
\end{lemma}



This result is analogous to the  proved for the \textsc{Split} algorithm in \cite{DBLP:conf/stoc/BuchbinderF25}, except for an additional $\varepsilon^{2} f(Q)$ term. 
The proof strategy also differs from that of \textsc{Split}. 
In the original analysis, the argument relies on the basis exchange property of matroids, which establishes a bijection between the selected set $T$ and any feasible set $Q$. 
Under a knapsack constraint, such a bijection no longer exists. 
Instead, we order the elements into a sequence and partition them into several segments so that they can be paired accordingly. 
As a result, at most one element $q \in Q$ with $f(q \mid T) \ge 0$ may remain unpaired. 
As shown above, the marginal contribution of such an element is at most $\varepsilon^{2} f(Q)$. 
We further show that this additional term is negligible in the overall analysis of the optimization algorithm.

\begin{algorithm}[H]
    
\caption{\textsc{KnapsackSplit} $(f, w, B, \ell)$}\label{alg:Split}
Initialize: $T_1 \gets \emptyset, T_2\gets\emptyset,\ldots, T_{\ell}\gets\emptyset$. \\
Use $T$ to denote the union $\cup_{j = 1}^{\ell} T_j$.\\
\While{$w(T) \le B$ }{
    Let $\mathcal{N}' \gets \{u \in N \backslash T | w(T \cup \{u\}) \le B \}$. \\
    Let $(u,j) \in arg\max_{(u,j)\in \mathcal{N'}\times[\ell]} \frac{f(u\mid T_j)}{w(u)}.$\\
    $T_j \gets T_j \cup \{u\}.$ \\
}
\Return $ (T_1, \cdots, T_{\ell}).$
\end{algorithm}

     \begin{algorithm}[H]
        \caption{\textsc{KnapsackDMCG($g$, $w$, $B$, $\eps$)}}
        \label{alg:deterministicMCG}
        Let $\delta \gets \eps^3$. \\
        Let $\mathbf{y}(0) \gets 0$. \\
        \For{i = 1 \KwTo $1/\delta$}{
            Define $g:2^{\mathcal{N}} \rightarrow \mathbb{R}_{\ge 0}$ to be $g(S) = F(\mathbf{e}_{S}\lor \mathbf{y}(t)).$\\
            $(T_1, \ldots, T_{1/\eps}) \gets $\textsc{Split}$(g,w,B,1/\eps)$ \\
            $\mathbf{y}(t+\delta) \gets \mathbf{y}(t) \oplus(\delta \cdot \sum_{j = 1}^{1/\epsilon}\mathbf{e}_{{T_j}})$.
        }
    \Return $\mathbf{y}(1)$.
    \end{algorithm}
\begin{proof}We first aim to establish a correspondence between the elements in $Q$ and those in 
$T = \bigcup_i T_i$. 
We arbitrarily order the elements in $Q$ as $q_1, q_2, \ldots, q_x$, 
and denote the elements of $T$ according to their order of added as 
$t_1, t_2, \ldots, t_y$, 
where $h_i$ indicates that the element $t_i$ was added to the set $T_{h_i}$. 
We additionally denote by $T_j^{(i)}$ the set of elements in $T_j$ at the moment when the element $t_i$ is added.

We represent each element as an interval on the real line $[0,B]$, with length equal to its weight. 
We first place the elements of $Q$ consecutively according to the fixed order: 
the element $q_1$ corresponds to the interval $[0, w(q_1)]$, 
$q_2$ corresponds to the interval $[w(q_1), w(q_1)+w(q_2)]$, 
and in general, $q_i$ corresponds to the interval
\[
I_i^{Q} = \Bigl[\sum_{k=1}^{i-1} w(q_k),\ \sum_{k=1}^{i} w(q_k)\Bigr].
\]

Similarly, we place the elements of $T=\bigcup_i T_i$ on the same interval $[0,B]$ according to their order of insertion. 
Specifically, the element $t_1$ corresponds to the interval $[0, w(t_1)]$, 
the element $t_2$ corresponds to the interval $[w(t_1), w(t_1)+w(t_2)]$, 
and in general, the element $t_i$ corresponds to the interval
\[
I_k^{T}  = \Bigl[\sum_{k =1}^{i-1} w(t_\ell),\ \sum_{k = 1}^{i} w(t_\ell)\Bigr].
\]


Based on the interval representations constructed above, we consider the pairwise intersections between intervals associated with elements in $Q$ and those associated with elements in $T$. 
Specifically, each point in $[0,B]$ lies in exactly one interval corresponding to an element of $Q$ and in at most one interval corresponding to an element of $T$, which naturally induces a pairing between elements of $Q$ and elements of $T$ through their overlapping intervals. 
For $\alpha$ indexing elements of $Q$ and $\beta$ indexing elements of $T$, we denote by $w_{\alpha,\beta}$ the length of the intersection between the corresponding intervals,
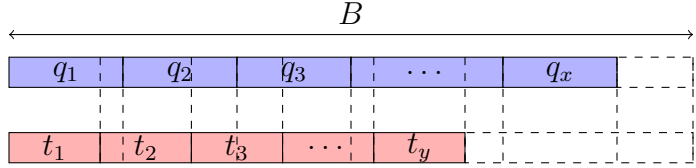
\[
w_{\alpha,\beta}
\;=\;
\bigl|\, I_\alpha^{Q} \cap I_\beta^{T} \,\bigr|.
\]

\begin{center}
\begin{tikzpicture}

\def\rectH{0.4} 
\def\Ytop{0} 

\draw[fill=blue!30] (0,\Ytop) rectangle (1.5,\Ytop+\rectH); \node at (0.75,\Ytop+0.2) {$q_1$};
\draw[fill=blue!30] (1.5,\Ytop) rectangle (3.0,\Ytop+\rectH); \node at (2.25,\Ytop+0.2) {$q_2$};
\draw[fill=blue!30] (3.0,\Ytop) rectangle (4.5,\Ytop+\rectH); \node at (3.75,\Ytop+0.2) {$q_3$};
\draw[fill=blue!30] (4.5,\Ytop) rectangle (6.5,\Ytop+\rectH); \node at (5.5,\Ytop+0.2) {$\dots$};
\draw[fill=blue!30] (6.5,\Ytop) rectangle (8.0,\Ytop+\rectH); \node at (7.25,\Ytop+0.2) {$q_x$};

\def\rectHbot{0.4} 
\pgfmathsetmacro{\Y}{\Ytop - 1}

\draw[fill=red!30] (0,\Y) rectangle (1.2,\Y+\rectHbot); \node at (0.6,\Y+0.2) {$t_1$};
\draw[fill=red!30] (1.2,\Y) rectangle (2.4,\Y+\rectHbot); \node at (1.8,\Y+0.2) {$t_2$};
\draw[fill=red!30] (2.4,\Y) rectangle (3.6,\Y+\rectHbot); \node at (3.0,\Y+0.2) {$t_3$};
\draw[fill=red!30] (3.6,\Y) rectangle (4.8,\Y+\rectHbot); \node at (4.2,\Y+0.2) {$\dots$};
\draw[fill=red!30] (4.8,\Y) rectangle (6.0,\Y+\rectHbot); \node at (5.4,\Y+0.2) {$t_y$};

\draw[dashed, fill=white, fill opacity=0] (8.0,\Ytop) rectangle (9.0,\Ytop+\rectH);
\draw[dashed, fill=white, fill opacity=0] (6.0,\Y) rectangle (9.0,\Y+\rectHbot);

\def\dashline{1.2, 1.5, 2.4, 3.0, 3.6, 4.5, 4.8, 6.0, 6.5, 8.0}
\foreach \i in \dashline {
    \draw[dashed] (\i,\Ytop + \rectH) -- (\i,\Y);
}
\draw[dashed] (9.0,\Ytop + \rectH) -- (9.0,\Y);

\draw[<->] (0, \Ytop+\rectH+0.3) -- (9.0, \Ytop+\rectH+0.3) node[midway, above] {$B$};

\end{tikzpicture}

\end{center}

For each $\alpha \in \{1,2,\ldots,x\}$ and $\beta \in \{1,2,\ldots,\mathbf{y}\}$, we have
\[
\frac{f(t_\beta \mid T_{h_\beta}^{(\beta-1)})}{w(t_\beta)}
\ge
\frac{1}{\ell} \sum_{j = 1}^{\ell}
\frac{f(q_\alpha \mid T_j^{(\beta-1)})}{w(q_\alpha)}
\ge
\frac{1}{\ell} \sum_{j = 1}^{\ell}
\frac{f(q_\alpha \mid T_j \cup Q_{\alpha - 1})}{w(q_\alpha)}.
\]
Multiplying both sides by $w_{\alpha,\beta}$ and summing over all
$\alpha \in \{1,\dots,x\}$ and $\beta \in \{1,\dots,\mathbf{y}\}$, we obtain
\[
\sum_{\beta=1}^{y} \sum_{\alpha=1}^{x}
\frac{w_{\alpha,\beta}}{w(t_\beta)}
f(t_\beta \mid T_{h_\beta}^{(\beta-1)})
\ge
\frac{1}{\ell} \sum_{j=1}^{\ell}
\sum_{\alpha=1}^{x} \sum_{\beta=1}^{y}
\frac{w_{\alpha,\beta}}{w(q_\alpha)}
f(q_\alpha \mid T_j \cup Q_{\alpha-1}).
\]
On the left-hand side, by construction we have
$\sum_{\alpha=1}^{x} w_{\alpha,\beta} \le w(t_\beta)$ for each $\beta$.
Thus, the left-hand side is lower bounded by
\[
\sum_{\beta=1}^{y} f(t_\beta \mid T_{h_\beta}^{(\beta-1)}) = \sum_{j = 1}^{\ell} f(T_j).
\]
On the right-hand side, the term $\sum_{\beta=1}^{y} w_{\alpha,\beta}$ admits a simple characterization depending on the position of $\alpha$ relative to $\gamma$. 
For all $\alpha < \gamma$, the interval $I_\alpha^{Q}$ is fully covered by the union of intervals corresponding to $T$, and hence $\sum_{\beta=1}^{y} w_{\alpha,\beta} = w(q_\alpha)$. 
For $\alpha = \gamma$, the interval $I_\alpha^{Q}$ is only partially covered, which implies $\sum_{\beta=1}^{y} w_{\alpha,\beta} \le w(q_\alpha)$. 
For $\alpha > \gamma$, the interval $I_\alpha^{Q}$ has empty intersection with the intervals corresponding to $T$ and lies entirely to their right, implying $w(q_\alpha) + w(T) \le B$. Thus,
\[
f(o_{\alpha} \mid T_i) \le 0, \quad \forall i = 1,2,\ldots,\ell,
\]
otherwise this element could be added to $T_i$. Summing and scaling, we obtain:
\[ 
\sum_{\alpha = \gamma + 1}^x  \sum_{\ell = 1}^{\ell}f(q_\alpha\mid T_j\cup Q_{\alpha- 1}) \le \sum_{\alpha = \gamma + 1}^x  \sum_{\ell = 1}^{\ell}f(q_\alpha\mid T_j) \le 0
 \]
Consequently, we have:

\begin{align*}
&\frac{1}{\ell} \sum_{j=1}^{\ell}
\sum_{\alpha=1}^{x} \sum_{\beta=1}^{y}
\frac{w_{\alpha,\beta}}{w(q_\alpha)}
f(q_\alpha \mid T_j \cup Q_{\alpha-1}) \\
= &\frac{1}{\ell}\sum_{\alpha=1}^{\gamma-1}
\sum_{j=1}^{\ell}
\sum_{\beta=1}^{y} \frac{w_{\alpha,\beta}}{w(q_\alpha)}
f(q_\alpha \mid T_j \cup Q_{\alpha-1}) 
+ \frac{1}{\ell}\sum_{j=1}^{\ell} \sum_{\beta=1}^{y} \frac{w_{\gamma,\beta}}{w(q_\gamma)}
f(q_\gamma \mid T_j \cup Q_{\gamma-1}) \\
=  & \frac{1}{\ell}\sum_{\alpha = 1}^{\gamma}\sum_{j=1}^{\ell}f(q_\alpha \mid T_j \cup Q_{\alpha-1})  - \left(1 - \sum_{\beta=1}^{y} \frac{w_{\gamma,\beta}}{w(q_\gamma)}\right)\frac{1}{\ell}\sum_{j=1}^{\ell}f(q_\gamma \mid T_j \cup Q_{\gamma-1}) \\
\ge & \frac{1}{\ell}\sum_{\alpha = 1}^{x}\sum_{j=1}^{\ell}f(q_\alpha \mid T_j \cup Q_{\alpha-1})  - \left(1 - \sum_{\beta=1}^{y} \frac{w_{\gamma,\beta}}{w(q_\gamma)}\right)\frac{1}{\ell}\sum_{j=1}^{\ell}f(q_\gamma \mid T_j \cup Q_{\gamma-1}) \\
\ge & \frac{1}{\ell}\sum_{\alpha = 1}^{x}\sum_{j=1}^{\ell}f(q_\alpha \mid T_j \cup Q_{\alpha-1})  - \left(1 - \sum_{\beta=1}^{y} \frac{w_{\gamma,\beta}}{w(q_\gamma)}\right)\frac{1}{\ell}\sum_{j=1}^{\ell}f(q_\gamma \mid \emptyset) \\
= & \frac{1}{\ell}\sum_{j=1}^{\ell}\left(f(Q \cup T_j) - f(T_j)\right)  - \eps^2 f(Q)
\end{align*}

\end{proof}

With the per-iteration progress guarantee in place, we can prove that the algorithm returns a continuous solution with a guaranteed approximation ratio.

\begin{theorem}
Given a submodular function $f: 2^{\mathcal N} \to \mathbb{R}^+ $, a knapsack constraint with capacity $B$ and a weight function 
$ w: \mathcal N \to \mathbb{R}^+ $, the $\textsc{KnapsackDMCG}$ algorithm returns a fractional solution $\mathbf{y} \in [0,1]^{2^\N} $ with $\sum_{u\in\N}\Mar_u(\mathbf{y}) w(u) \le B$, and 

\begin{equation}
\label{eq:knapsack_split_approx}
    F(\mathbf{y}(1)) \ge  (\frac{1}{e} - O(\eps)) f(\OPT)
\end{equation} 
where $F$ is the extended multilinear extension of $f$.
    Besides, for all $i \in [0 , 1/\delta]$, we have $\operatorname{frac}(\mathbf{y}(i\delta)) \le 1 / \eps \cdot i \le \varepsilon^{-4}$.

\end{theorem}

\begin{proof}

Firstly, we briefly explain why \eqref{eq:knapsack_split_approx} holds, as this part of the proof is identical to the argument in Section~3 (the matroid case) and can be directly derived from~\cite{DBLP:conf/stoc/BuchbinderF25}. By submodularity and the choice $\ell = 1/\varepsilon$, we obtain
\[
\sum_{j=1}^{1/\varepsilon} g(T_j \mid \varnothing)
\;\ge\;
(1-O(\varepsilon))\,g(\OPT)
-
\varepsilon \sum_{j=1}^{1/\varepsilon} g(T_j).
\]
Plugging this inequality into the iterative process yields
\[
\sum_{j=1}^{1/\varepsilon}
\left(
F(\mathbf e_{T_j}\vee \mathbf{y}((i-1)\delta))
-
F(\mathbf{y}((i-1)\delta))
\right)
\;\ge\;
(1-O(\varepsilon))\cdot
F(\mathbf e_{\OPT}\vee \mathbf{y}((i-1)\delta))
-
(1-\varepsilon)\cdot
F(\mathbf{y}((i-1)\delta)).
\]
Expanding the update rule of the algorithm gives
\[
\frac{1}{\delta}
\left(
F(\mathbf{y}(i\delta)) - F(\mathbf{y}((i-1)\delta))
\right)
\;\ge\;
(1-O(\varepsilon))\cdot
F(\mathbf e_{\OPT}\vee \mathbf{y}((i-1)\delta))
-
F(\mathbf{y}((i-1)\delta)).
\]
Finally, using $ \| \operatorname{Mar}(\mathbf{y}(i\delta)) \| \le  1 - (1 - \delta)^i$  together with Observation~\ref{obs:eme} we can acquire
\[
F(\mathbf{e}_{\OPT} \vee \mathbf{y}((i-1)\delta)) \ge (1-\delta)^{i-1}\cdot f(\OPT),
\]
and
solving this recurrence over the interval $[0,1]$ yields
\[
F(\mathbf{y}(1))
\;\ge\;
\left(\frac{1}{e}-O(\varepsilon)\right)f(\OPT),
\]
which establishes \eqref{eq:knapsack_split_approx}.

Besides,it is easy to see that in each iteration, when the algorithm updates $\mathbf{y}(t)$ to $\mathbf{y}(t+\delta)$, it increases at most $1/\varepsilon$ coordinates. 
Since the algorithm performs at most $1/\varepsilon^{3}$ iterations in total, we have
$\operatorname{frac}(\mathbf{y}(i\delta)) \le (1/\varepsilon)\cdot i$. Now it suffices to show that
\[
\sum_{u \in \N} \Mar_u(\mathbf{y})\, w(u) \le B .
\]
Let $T_j^{(i)}$ denote the set $T_i$ produced in the $j$-th iteration of the algorithm. 
The final vector $\mathbf{y}$ can be written as
\[
\mathbf{y} = \bigoplus_{j = 1}^{1/\delta} \Bigl( \delta \cdot \sum_{i = 1}^{1/\varepsilon} \mathbf{e}_{T_j^{(i)}} \Bigr).
\]
For each fixed $j$, the \textsc{KnapsackSplit} algorithm guarantees that the total weight of the selected sets does not exceed the capacity, implying
\[
\Mar\Bigl(\delta \cdot \sum_{i = 1}^{1/\varepsilon} \mathbf{e}_{T_j^{(i)}}\Bigr)
= \delta \sum_{i = 1}^{1/\varepsilon} w(T_j^{(i)}) \le \delta B .
\]
By the additivity of the operator $\Mar(\cdot)$ over the direct sum, we obtain
\[
\Mar(\mathbf{y})
= \bigoplus_{j = 1}^{1/\delta}
\Mar\Bigl(\delta \cdot \sum_{i = 1}^{1/\varepsilon} \mathbf{e}_{T_j^{(i)}}\Bigr)
\le \sum_{j = 1}^{1/\delta} \delta B
= B .
\]
\end{proof}

 \begin{algorithm}[H]
        \caption{\textsc{Relax}($\mathbf{y}, u$)}
        \label{alg:relax}
        Update $\mathbf{y}_{u} \gets \operatorname{Mar}_u(\mathbf{y})$\\
        \For{every non-empty set $S \subseteq \mathcal{N} - u$}{
            $\mathbf{y}_S \gets 1 - (1 - \mathbf{y}_{S+u})(1 - \mathbf{y}_S).$\\
            $\mathbf{y}_{S+u} \gets 0.$
        }
    \Return $\mathbf{y}$.
    \end{algorithm}

\begin{algorithm}
\caption{\textsc{Rounding}($f, w, \mathbf{y}$)}

\While{$R \neq \NN$}{
   $S \gets R \cap \{u \in \mathcal N' \mid \mathbf{y}_{\{u\}} \in (0,1)\}$.\\
    \While{$|S| < 2$}{
        Choose $u \in \mathcal N' \setminus R$ .\\
        $\mathbf{y} \gets \textsc{Relax}(\mathbf{y}, u)$.\\
        $R \gets R \cup \{u\}$.\\
        Update $S \gets R \cap \{u \in \mathcal N' \mid \mathbf{y}_{\{u\}} \in (0,1)\}$.
    }

    Pick distinct $u, v \in S$.\\
    Define
    $g(t) \coloneqq
 F\!\left(\mathbf{y} + t\!\left(\frac{\mathbf e_{\{u\}}}{w(u)} - \frac{\mathbf e_{\{v\}}}{w(v)}\right)\right)$.\\

    Let
    $t_{\max} = \min\{ (1 - \mathbf{y}_{\{u\}}) w(u),\; \mathbf{y}_{\{v\}} w(v) \}$ and
    $t_{\min} = \max\{ -\mathbf{y}_{\{u\}} w(u),\; (\mathbf{y}_{\{v\}} - 1) w(v) \}$.\\

    Let $t^\star \in \arg\max_{t \in \{t_{\min}, t_{\max}\}} g(t)$.\\
    $\mathbf{y} \gets \mathbf{y} + t^\star\!\left(\frac{\mathbf e_{\{u\}}}{w(u)} - \frac{\mathbf e_{\{v\}}}{w(v)}\right)$.
}
\Return $\arg\max\Bigl\{\, f(R),\; f\!\bigl(R \setminus \{\,u \in \mathcal N \mid \mathbf{y}_{\{u\}} \in (0,1)\,\}\bigr) \Bigr\}$.

\end{algorithm}

\subsection{Rounding}
Finally, we apply a rounding algorithm to the fractional solution $\mathbf{y}$ obtained by running the optimization algorithm on the function $f(\cdot \cup E)$.

Our rounding algorithm is inspired by the Pipage Rounding paradigm and employs the \textsc{Relax} algorithm of \cite{DBLP:conf/stoc/BuchbinderF25} as a key subroutine. 
The \textsc{Relax} algorithm transforms a vector $\mathbf{x} \in [0,1]^{2^{n}}$ by aggregating all components corresponding to sets that contain a given element $u$ into the coordinate $\mathbf{e}_u$. 
This transformation preserves the marginal vector $\mathrm{Mar}(\mathbf{x})$, increases the support size by at most one, and does not decrease the objective value.

\begin{lemma}[Lemma 5.1 in \cite{DBLP:conf/stoc/BuchbinderF25}]
    Let $\mathbf{x} \in [0, 1]^{2^\mathcal{N}}$, $u \in \mathcal{N}$ and $\mathbf{z} = \textsc{Relax}(\mathbf{x}, u) \in [0, 1]^{2^\mathcal{N}}$. Computing $\mathbf{z}$ requires $O(\mathrm{frac}(\mathbf{x}))$ time, and the new vector satisfies:

\begin{itemize}
    \item $\mathrm{frac}(\mathbf{z}) \le 1 + \mathrm{frac}(\mathbf{x})$.
    \item $\mathrm{Mar}(\mathbf{x}) = \mathrm{Mar}(\mathbf{z})$, and $x_S = 0$ for all sets $S$ that contain $u$, but are not the singleton set $\{u\}$.
    \item $F(\mathbf{z}) \ge F(\mathbf{x})$.
\end{itemize}
\end{lemma}

Our rounding procedure iteratively reduces the number of fractional elements. 
At each step, we consider the elements whose corresponding coordinates in $\mathbf{y}$ are fractional after applying \textsc{Relax}. If fewer than two such elements remain, we \textsc{Relax} additional elements so that exactly two fractional elements are present, denoted by $u$ and $v$. 
We then examine the extended multilinear extension along the direction
\[
g(t) = F\Bigl(\mathbf{y} + t\Bigl(\frac{\mathbf{e}_{\{u\}}}{w(u)} - \frac{\mathbf{e}_{\{v\}}}{w(v)}\Bigr)\Bigr).
\]
We can see that $g$ is convex.

\begin{lemma}[Convexity along exchange directions]
\label{lem:exchange_convexity}
Let $f : 2^{\mathcal N} \to \mathbb{R}$ be a submodular function and let
$F$ denote its multilinear extension.
For any $\mathbf{y} \in [0,1]^{\mathcal N}$ and any two distinct elements
$u,v \in \mathcal N$, the function
\[
g(t)
=
F\Bigl(
\mathbf{y} + t\Bigl(
\frac{\mathbf e_{\{u\}}}{w(u)}
-
\frac{\mathbf e_{\{v\}}}{w(v)}
\Bigr)
\Bigr)
\]
is convex in $t$ over its feasible interval.
\end{lemma}
\begin{proof}
It's easy to know that 
\[
\begin{aligned}
g''(t)
&=
\frac{1}{w(u)^2}
\left.
\frac{\partial^2 F(z)}{\partial z_{\{u\}}^2}
\right|_{z=\mathbf{y}+t\left(\frac{\mathbf e_{\{u\}}}{w(u)}-\frac{\mathbf e_{\{v\}}}{w(v)}\right)}+
\frac{1}{w(v)^2}
\left.
\frac{\partial^2 F(z)}{\partial z_{\{v\}}^2}
\right|_{z=\mathbf{y}+t\left(\frac{\mathbf e_{\{u\}}}{w(u)}-\frac{\mathbf e_{\{v\}}}{w(v)}\right)}
\\[6pt]
&\quad-
\frac{2}{w(u)w(v)}
\left.
\frac{\partial^2 F(z)}{\partial z_{\{u\}} \partial z_{\{v\}}}
\right|_{z=\mathbf{y}+t\left(\frac{\mathbf e_{\{u\}}}{w(u)}-\frac{\mathbf e_{\{v\}}}{w(v)}\right)} .
\end{aligned}
\]
By Observation \ref{observation 8} we have 
\[
g''(t)
=
-\frac{2}{w(u)w(v)}
\left.
\frac{\partial^2 F(z)}{\partial z_{\{u\}}\partial z_{\{v\}}}
\right|_{z=\mathbf{y}+t\left(\frac{\mathbf e_{\{u\}}}{w(u)}-\frac{\mathbf e_{\{v\}}}{w(v)}\right)} \ge 0.
\]
\end{proof}





Our algorithm then updates $\mathbf{y}$ to the vector corresponding to the endpoint of $g$ that achieves the larger function value. This procedure is repeated iteratively until all elements have been processed by \textsc{Relax}, leaving at most one fractional coordinate. At this point, the algorithm compares the two candidate sets: one consisting of all elements rounded to 1, and the other consisting of these elements together with the remaining fractional element, and selects the set with the larger function value. The rounding algorithm comes with the following performance guarantee:

\begin{theorem}
\label{thm:knapsack_rounding}
Given a submodular function $f : 2^{\mathcal N} \to \mathbb{R}^+$ , a knapsack constraint with $w(u)\le \eps B$ for all $u \in \mathcal{N}$ and a fractional solution $\mathbf{y}$ satisfying
$
\sum_{u \in \mathcal N} \mathrm{Mar}_u(\mathbf{y})\, w(u) \le B,
$ our \textsc{Rounding} algorithm returns a discrete solution $S \subseteq \mathcal N$ such that
\[
f(S) \ge F(\mathbf{y})
\quad \text{and} \quad
w(S) \le (1+\varepsilon)B,
\]
and the algorithm uses $O(n\cdot 2^{\operatorname{frac}(\mathbf{y})}) $ queries.
\end{theorem}

\begin{proof}
We first analyze the behavior of the algorithm when $\mathbf{y}$ is updated, specifically during lines 9–12.
 By convexity, one of the two endpoints of this interval achieves a value at least as large as the current one, allowing us to round either $u$ or $v$ to an integral value without decreasing the objective. Moreover, moving along this direction preserves the weighted sum
\[
\sum_{u \in \mathcal N} \mathrm{Mar}_u(\mathbf{y})\cdot w(u),
\] 
Therefore, the knapsack capacity is maintained.

During the execution, applying \textsc{Relax} may increase $\operatorname{frac}(\mathbf{y})$ by 1, which also increases $|S|$ by 1. Each step along a convex exchange direction decreases both $\operatorname{frac}(\mathbf{y})$ and $|S|$ by 1. Since $|S| \le 2$ is always maintained, it follows that $\operatorname{frac}(\mathbf{y})$ can increase by at most 2 relative to its initial value.

At the final iteration, it is possible that only a single element remains fractional. In this case, the current objective is a convex combination of the values obtained by either selecting or discarding this element. Hence, we can choose the better integral option, ensuring
\[
f(S) \ge F(\mathbf{y}).
\]

Finally, since each rounding step preserves the weighted sum and only one element may remain fractional at the end, we have
\[
w(S) \le (1+\varepsilon)B.
\]

The total number of queries is $O(n \cdot 2^{\operatorname{frac}(\mathbf{y})})$, as each evaluation of $F$ requires $2^{\operatorname{frac}(\mathbf{y})}$ queries and there are at most $n$ elements to process.
\end{proof}
This final decision may cause the knapsack capacity to be exceeded; however, the violation is bounded by at most $\varepsilon B$. 
Since the rounding procedure is executed with an initial capacity of $(1-\varepsilon)B$, such a bounded capacity violation can be safely tolerated.


\subsection{Proof of the Main Result}
In light of the above descriptions, we are now ready to combine all the preceding results to prove Theorem~\ref{knapsack_main}.

\begin{proof}[Proof of Theorem ~\ref{knapsack_main}]

Algorithm~\ref{Knapsack:whole}  first enumerates all sets \(E_i\) with \(|E_i| \le \varepsilon^{-2}\), which ensures that the set constructed in Lemma~2, denoted by \(E\), is included among the enumerated candidates. At this point, our algorithm then applies an optimization algorithm to the resulting problem instance with objective function \(g(S) = f(S \cup E_i)\) and knapsack capacity \((1 - \eps)(B - w(E))\).
Then we obtain a vector $\mathbf{y}$ for $G(\cdot)$, where $G$ is the extended multilinear extension of $g$. Due to the property of $E$, this guarantees that
\[
G(\mathbf{y}) \ge \left(\frac{1}{e} - O(\eps)\right) g(\OPT) 
\ge \left(\frac{1}{e} - \eps\right) f(\OPT),
\]
and
\[
\sum_{u \in \mathcal N} \mathrm{Mar}_u(\mathbf{y})\cdot w(u) 
\le (1 - \eps)(B - w(E)).
\]

We then apply a rounding procedure, which returns two discrete solutions $R$ and $R + u$ such that
\[
\max\{g(R), g(R+u)\} \ge g(\mathbf{y})
\]
and
\[
w(R+u) \le B - w(E).
\]
We select the one with the larger value and denote it by $S$.

Finally, our algorithm outputs the set $S_i \cup E_i$ with the largest function value among all candidates, which in particular includes the set $S \cup E$ discussed above. Since $w(S \cup E) \le B$ and 
\[
f(S \cup E) = g(S) \ge g(\mathbf{y}) \ge \left(\frac{1}{e} - \eps\right) f(\OPT),
\]
both the feasibility and the approximation ratio are satisfied.

For the query complexity, the enumeration step first incurs a multiplicative overhead of $2^{O(\eps^{-2})}$.

In the optimization phase, we always maintain $\operatorname{frac}(\mathbf{y}) \le \eps^{-4}$. The process runs for $\eps^{-3}$ rounds, and in each round the algorithm evaluates the function for $n^2$ candidate directions. Since querying the value of $g(\mathbf{y})$ requires $2^{\operatorname{frac}(\mathbf{y})}$ value queries to $g$, and $\operatorname{frac}(\mathbf{y}) \le \eps^{-4}$ throughout the algorithm, each such evaluation costs at most $2^{\eps^{-4}}$ queries. Therefore, the total number of queries in this phase is $O( n^2 \cdot  2^{\eps^{-4}})$. Then, the rounding algorithm takes a fractional solution with $\operatorname{frac}(\mathbf{y}) \le \eps^{-4}$ as input.Each evaluation requires $2^{\operatorname{frac}(\mathbf{y})} \le 2^{\eps^{-4}}$ queries, resulting in $O(n \cdot 2^{\eps^{-4}})$ queries in total.
Combining all parts, the overall number of value queries is
\[
n^{O(\eps^{-2})} \cdot \left(O( n^2 \cdot 2^{\eps^{-4}}) + O(n \cdot 2^{\eps^{-4}})\right) = O_{\eps}(n^{\eps^{-2}}).
\]
\end{proof}

\section{Conclusion}

In this work, based on the optimization--then--rounding paradigm over the extended multilinear extension, we design deterministic algorithms for maximizing non-monotone submodular functions under matroid and knapsack constraints. Our algorithms achieve improved approximation ratios compared with the best previously known deterministic results.

For the future, the extended multilinear extension framework still has broader potential for further study, although many related questions remain challenging. One possible direction is to derandomize the aided continuous greedy algorithm for knapsack constraints, which would lead to a deterministic algorithm with an approximation ratio of $0.385$. Another direction is to extend the constraint to a multi-dimensional knapsack. We have verified that in the optimization phase, the multiplicative updates technique can be applied within the extended multilinear extension framework. Nevertheless, a deterministic rounding algorithm tailored for the multi-dimensional knapsack constraint is still required. Such a rounding procedure must simultaneously maintain feasibility for all knapsack constraints while carefully controlling the number of fractional coordinates in the vector.

\section*{Acknowledgments}
This work was supported in part by the National Natural Science Foundation of China Grants Nos. 62325210, 12501450, 62272441.  


\bibliographystyle{abbrv}
\bibliography{deterministic}



\end{document}